\documentclass[a4paper,pdflatex]{llncs}
\usepackage[utf8]{inputenc}
\usepackage{amsmath,amssymb}
\usepackage{theorem}
\usepackage{algorithmicx}
\usepackage[ruled,vlined,titlenumbered,linesnumbered]{algorithm2e}
\usepackage{mathrsfs}
\usepackage[table]{xcolor}
\usepackage[colorlinks=true,citecolor=blue,linkcolor=blue,urlcolor=blue,breaklinks=true]{hyperref}
\usepackage{bm}
\usepackage{xspace}
\usepackage{scalerel}
\usepackage{cite}
\usepackage{url}
\usepackage{array}
\usepackage{pgfplots}
\usetikzlibrary{plotmarks,patterns}
\usepackage[nolist]{acronym}
\usepackage{caption}
\usepackage{subcaption}
\usepackage{orcidlink}
\usepackage{tabto}

\let\normalunderbrace=\underbrace
\usepackage{mathtools}
\let\underbrace=\normalunderbrace
\mathtoolsset{showonlyrefs}

\usepackage[margin=3.8cm]{geometry}

\pgfplotsset{
  compat=1.17,
    every axis/.append style={
        scale only axis,
  width=0.55\columnwidth,
  height=0.4\columnwidth,
  label style={inner sep=0, font=\normalsize},
  tick label style={font=\scriptsize},
  legend style={font=\scriptsize},
  mark size=3,
  major grid style={dashed},
  line width=0.8pt,
  axis line style = thin}
}

\tikzstyle{SB}    = [color=black, solid]
\tikzstyle{LRS} = [color=red, dash pattern=on 2pt off 4pt on 6pt off 4pt, mark=x, mark options={solid}]
\tikzstyle{FLRS}   = [color=blue, dashed, mark=diamond, mark options={solid}]
\tikzstyle{HRFLRS} = [color=teal, dashed, mark=o, mark options={solid}]
\tikzstyle{divergence_1} = [color=purple]
\tikzstyle{divergence_2} = [color=teal]
\tikzstyle{divergence_3} = [color=blue]
\tikzstyle{distribution} = [color=black, mark options={scale=.1, draw=black, fill=black}, mark=*]

\newcolumntype{M}[1]{>{\centering\arraybackslash}m{#1}}

\newcommand{\Fqm}{\ensuremath{\mathbb F_{q^m}}}

\newcommand{\Fqmh}{\ensuremath{\mathbb F_{q^{mh}}}}

\newcommand{\Fqd}{\ensuremath{\mathbb F_{q^d}}}
\newcommand{\Fq}{\ensuremath{\mathbb F_{q}}}

\newcommand{\F}{\ensuremath{\mathbb F}}
\newcommand{\ZZ}{\ensuremath{\mathbb{Z}}}
\newcommand{\NN}{\ensuremath{\mathbb{N}}}

\newcommand{\set}[1]{\ensuremath{\mathcal{#1}}}

\newcommand{\Polyring}{\ensuremath{\Fqm[x]}}

\newcommand{\aut}{\ensuremath{\sigma}}
\newcommand{\der}{\ensuremath{\delta}}
\newcommand{\derPar}{\ensuremath{z}}
\newcommand{\frob}{\ensuremath{\theta}}
\newcommand{\frobPar}{\ensuremath{u}}
\newcommand{\SkewPolyring}{\ensuremath{\Fqm[x;\aut,\der]}}
\newcommand{\SkewPolyringZeroDer}{\ensuremath{\Fqm[x;\aut]}}

\newcommand{\MultSkewPolyring}{\ensuremath{\Fqm[x,y_1,\dots,y_\intOrder;\aut,\der]}}

\newcommand{\remev}[2]{{#1}\!\left[#2\right]}
\newcommand{\opev}[3]{\ensuremath{{#1}(#2)_{#3}}}

\newcommand{\op}[2]{\ensuremath{\mathcal{D}_{#1}(#2)}}
\newcommand{\opexp}[3]{\ensuremath{\mathcal{D}_{#1}^{#3}(#2)}}
\newcommand{\conj}[2]{\ensuremath{{#1}^{#2}}}

\newcommand{\algoname}[1]{{\normalfont\textsc{#1}}}

\newcommand{\ext}{\ensuremath{\text{ext}}}
\newcommand{\extInv}{\ensuremath{\text{ext}^{-1}}}

\newcommand{\piInv}{\pi^{-1}}

\makeatletter
\patchcmd{\@algocf@start}
  {-1.5em}
  {0pt}
{}{}
\makeatother

\newcommand{\OCompl}[1]{\ensuremath{\mathcal{O}({#1})}}

\newcommand{\defeq}{:=}

\renewcommand{\bar}{\overline}

\newcommand{\modr}{\; \mathrm{mod}_\mathrm{r} \;}

\DeclareMathOperator{\divides}{|}

\DeclareMathOperator{\Id}{\textrm{Id}}

\DeclareMathOperator{\wt}{wt}
\DeclareMathOperator{\rk}{rk}

\DeclareMathOperator{\diag}{diag}
\DeclareMathOperator{\unif}{unif}

\DeclareMathOperator{\RCEF}{RCEF}

\DeclareMathOperator{\lcm}{lcm}

\newcommand{\mystack}[2]{\ensuremath{\genfrac{}{}{0pt}{}{#1}{#2}}}

\renewcommand{\vec}[1]{\ensuremath{\bm{#1}}}
\newcommand{\mat}[1]{\ensuremath{\bm{#1}}}

\newcommand{\IPrem}[1]{\mathcal{I}_{#1}^{\mathrm{rem}}}

\newcommand{\opVandermonde}[3]{\ensuremath{\mathfrak{m}_{#1}(#2)_{#3}}}
\newcommand{\opMoore}[3]{\ensuremath{\mathfrak{M}_{#1}(#2)_{#3}}}
\newcommand{\genNorm}[2]{\ensuremath{\mathcal{N}_{#1}(#2)}}

\newcommand{\lclm}{\ensuremath{\mathrm{lclm}}}

\renewcommand{\a}{\mathbf a}
\renewcommand{\b}{\mathbf b}
\renewcommand{\c}{\mathbf c}
\renewcommand{\d}{\mathbf d}

\newcommand{\f}{\mathbf f}

\newcommand{\m}{\mathbf m}
\newcommand{\n}{\mathbf n}
\newcommand{\p}{\mathbf p}
\newcommand{\q}{\mathbf q}

\renewcommand{\t}{\mathbf t}

\renewcommand{\v}{\mathbf v}

\newcommand{\x}{\mathbf x}
\newcommand{\y}{\mathbf y}
\newcommand{\z}{\mathbf z}

\newcommand{\A}{\mathbf A}
\newcommand{\B}{\mathbf B}
\newcommand{\C}{\mathbf C}

\newcommand{\E}{\mathbf E}

\renewcommand{\L}{\mathbf L}

\newcommand{\N}{\mathbf N}
\renewcommand{\P}{\mathbf P}

\newcommand{\R}{\mathbf R}
\renewcommand{\S}{\mathbf S}

\newcommand{\X}{\mathbf X}
\newcommand{\Y}{\mathbf Y}

\newcommand{\0}{\mathbf 0}
\newcommand{\vecalpha}{\ensuremath{\boldsymbol{\alpha}}}
\newcommand{\vecbeta}{\ensuremath{\boldsymbol{\beta}}}
\newcommand{\vecgamma}{\ensuremath{\boldsymbol{\gamma}}}
\newcommand{\veczeta}{\ensuremath{\boldsymbol{\zeta}}}

\newcommand{\mycode}[1]{\ensuremath{\mathcal{#1}}}

\newcommand{\skewRS}[1]{\ensuremath{\mathrm{SRS}[#1]}}

\newcommand{\linRS}[1]{\ensuremath{\mathrm{LRS}[#1]}}

\newcommand{\foldedLinRS}[1]{\ensuremath{\mathrm{F}\mathrm{LRS}[#1]}}
\newcommand{\foldedSkewRS}[1]{\ensuremath{\mathrm{FSRS}[#1]}}

\newcommand{\HammingWeight}{\ensuremath{\wt_{H}}}

\newcommand{\SumRankWeight}{\ensuremath{\wt_{\Sigma R}}}
\newcommand{\SkewWeight}{\ensuremath{\wt_{skew}}}

\newcommand{\SumRankDist}{d_{\ensuremath{\Sigma}R}}

\newcommand{\SkewDist}{d_{skew}}

\newcommand{\softoh}[1]{\bnd{\widetilde{\mathcal{O}}}{#1}}

\newcommand{\bnd}[2]{\ensuremath{#1\mathopen{}\left(#2\right)\mathclose{}}}

\newcommand{\OMul}[1]{\mathcal{M}(#1)}

\newcommand{\shot}[2]{\ensuremath{{#1}^{(#2)}}}

\newcommand{\pe}{\ensuremath{\alpha}}

\newcommand{\degConstraint}{\ensuremath{D}}
\newcommand{\intOrder}{\ensuremath{s}}

\newcommand{\foldPar}{\ensuremath{h}}
\newcommand{\foldParShot}[1]{\ensuremath{\foldPar_{#1}}}
\newcommand{\foldParVec}{\ensuremath{\h}}

\newcommand{\foldOp}[1]{\mathcal{F}_{#1}}
\newcommand{\foldOpInv}[1]{\mathcal{F}_{#1}^{-1}}

\newcommand{\intDim}{\ensuremath{s}}
\newcommand{\shots}{\ensuremath{\ell}}

\newcommand{\lenFLRS}{\ensuremath{N}}
\newcommand{\lenFLRSVec}{\ensuremath{\N}}
\newcommand{\lenFLRSshot}[1]{\ensuremath{\lenFLRS_{#1}}}
\newcommand{\len}{\ensuremath{n}}
\newcommand{\lenVec}{\ensuremath{\n}}
\newcommand{\lenShot}[1]{\ensuremath{\len_{#1}}}

\newcommand{\height}{o}
\newcommand{\heightShot}[1]{\height_{#1}}
\newcommand{\heightVec}{\mathbf \height}

\newcommand{\h}{\vec{h}}

\newcommand{\ispecial}{j}
\newcommand{\ispecialTwo}{y}

\newcommand{\FSRSoffset}{\ensuremath{\omega}}

\begin{acronym}
 \acro{BMD}{bounded minimum distance}
 \acro{SRS}{skew Reed--Solomon}
 \acro{ISRS}{interleaved skew Reed--Solomon}
 \acro{FSRS}{folded skew Reed--Solomon}
 \acro{FRS}{folded Reed--Solomon}
 \acro{lclm}{least common left multiple}
 \acro{RS}{Reed--Solomon}
 \acro{LRS}{linearized Reed--Solomon}
 \acro{LLRS}{lifted linearized Reed--Solomon}
 \acro{ILRS}{interleaved linearized Reed--Solomon}
 \acro{LILRS}{lifted interleaved linearized Reed--Solomon}
 \acro{FLRS}{folded linearized Reed--Solomon}
 \acro{MSRD}{maximum sum-rank distance}
 \acro{KL}{Kullback--Leibler}
\end{acronym}

\begin{document}

\title{Interpolation-Based Decoding of Folded Variants \\ of Linearized and Skew Reed--Solomon Codes}

\author{
    Felicitas Hörmann$^{1,2}$\,\orcidlink{0000-0003-2217-9753} \and
    Hannes Bartz$^{1}$\,\orcidlink{0000-0001-7767-1513}\\
    \email{$\{$felicitas.hoermann, hannes.bartz$\}$@dlr.de}
}

\institute{
    \centering
    $^1$Institute of Communications and Navigation\\
    German Aerospace Center (DLR)\\
    Oberpfaffenhofen-Wessling, Germany
    \\[10pt]
    \centering
    $^2$School of Computer Science\\
    University of St. Gallen\\
    St. Gallen, Switzerland
}

\maketitle

\begin{abstract}
    The sum-rank metric is a hybrid between the Hamming metric and the rank metric and suitable for error correction in multishot network coding and distributed storage as well as for the design of quantum-resistant cryptosystems.
    In this work, we consider the construction and decoding of \ac{FLRS} codes, which are shown to be \ac{MSRD} for appropriate parameter choices.
    We derive an efficient interpolation-based decoding algorithm for \ac{FLRS} codes that can be used as a list decoder or as a probabilistic unique decoder.
    The proposed decoding scheme can correct sum-rank errors beyond the unique decoding radius with a computational complexity that is quadratic in the length of the unfolded code.
    We show how the error-correction capability can be optimized for high-rate codes by an alternative choice of interpolation points.
    We derive a heuristic upper bound on the decoding failure probability of the probabilistic unique decoder and verify its tightness by Monte Carlo simulations.
    Further, we study the construction and decoding of \acl{FSRS} codes in the skew metric.
    Up to our knowledge, \ac{FLRS} codes are the first \ac{MSRD} codes with different block sizes that come along with an efficient decoding algorithm.
\end{abstract}



\acresetall

\section{Introduction} \label{sec:introduction}

The sum-rank metric was first considered in~\cite{lu2005unified} in the context of space-time coding and covers the
Hamming metric as well as the rank metric as special cases.
Alternative decoding metrics as the sum-rank metric are of great interest to the field of code-based
cryptography (see e.g.~\cite{puchinger2020generic}).
Other applications range from error control in multishot network coding as described in~\cite{martinez2019reliable}
and~\cite{nobrega2010multishot} to the construction of locally repairable codes~\cite{martinez2019universal}.

Mart{\'\i}nez-Pe{\~n}as introduced \acf{LRS} codes, which generalize the well-studied families of \ac{RS} and Gabidulin
codes, in~\cite{martinez2018skew}.
\ac{LRS} codes are \acf{MSRD} codes, that is their minimum distance achieves the Singleton-like bound with equality.

While codewords of sum-rank-metric codes are commonly defined as tuples containing matrices of arbitrary sizes, most known constructions
use the same number of rows for every matrix in the tuple.
Some examples of \ac{MSRD} codes with different numbers of rows for the matrices can be found
in~\cite{byrne2021fundamental,camps2022optimal}.
Another construction for \ac{MSRD} codes with this property is given in~\cite{chen2022linear,chen2022new}.
However, no efficient decoding algorithm has been developed for such codes up to our knowledge.
We address this by presenting the family of \ac{FLRS} codes along with an efficient interpolation-based decoding
algorithm that can be used for list and probabilistic unique decoding.

We further apply our results to the skew-metric regime where we fold \ac{SRS} codes.
\ac{SRS} codes were introduced and studied in~\cite{boucher2014linear,liu2015construction} with respect to
Hamming metric and rank metric.
The work~\cite{martinez2018skew} defined the skew metric and analyzed \ac{SRS} codes in this new context.
In fact, it was shown that the sum-rank metric and \ac{LRS} codes are the linearized version of the skew metric and
\ac{SRS} codes, respectively.

The idea of folding constructions in coding evolved in the Hamming-metric context with Parvaresh--Vardy
codes~\cite{parvaresh2005correcting} and \acl{FRS} codes~\cite{Guruswami2008Explicit,brauchle2015}.
Folded Gabidulin codes and their efficient decoding in the rank metric were studied
in~\cite{bartz2017algebraic,BartzSidorenko_FoldedGabidulin2015_DCC,bartz2015list}.

\paragraph{Contributions}
Note that parts of this work were presented at WCC 2022: The Twelfth International Workshop on Coding and Cryptography
(see~\cite{hoermann2022efficient}).

We define the family of \ac{FLRS} codes and derive an interpolation-based decoding framework for these codes.
In contrast to~\cite{hoermann2022efficient}, we allow
different block sizes in the underlying unfolded code as well as the usage of different folding parameters.
This yields codes whose codewords are matrix tuples consisting of matrices having not the same size.
We further lift the restriction to skew polynomials with zero derivation and also deal with nonzero derivations.

As in~\cite{hoermann2022efficient}, we show how the decoding scheme can be used for list and probabilistic unique
decoding and give bounds on the list size and the failure probability, respectively.
We have performed several Monte Carlo simulations that verify the heuristic upper bound on the failure probability
empirically.
Moreover, new simulations show the reasonability of an assumption which is needed to obtain the heuristic bound.

A Justesen-like approach, which yields an improved interpolation-based decoding scheme for high-rate \ac{FLRS} codes,
and the discussion of implications for the skew metric are completely new topics in this work.
More precisely, we introduce \ac{FSRS} codes in the skew metric in a similar fashion as \ac{FLRS} codes and show how the proposed \ac{FLRS} decoding scheme can be applied.

\paragraph{Outline}
We start by giving some preliminaries in~\autoref{sec:preliminaries} before defining \ac{FLRS} codes and studying
their minimum distance in~\autoref{sec:flrs-codes}.

The main part of this paper is~\autoref{sec:decoding} in which we present and investigate an interpolation-based
decoding scheme for \ac{FLRS} codes.
The decoder consists of an interpolation step and a root-finding step which are explained in detail
in~\autoref{subsec:interpolation-step} and~\autoref{subsec:root-finding-step}, respectively.
\autoref{subsec:list-and-probabilistic-decoding} shows how the presented scheme can be used for list and probabilistic
unique decoding.
In particular, we derive an upper bound on the list size in the first case and on the failure probability in the latter.
\autoref{subsec:improved_high_rate} introduces a variant of the decoding scheme that is tailored to high-rate codes by
using a different set of interpolation points.
Since the bound on the failure probability for probabilistic unique decoding
from~\autoref{subsec:list-and-probabilistic-decoding} is heuristic, we empirically verify its validity by simulations in
SageMath in~\autoref{subsec:simulation_results}.

\autoref{sec:implications-fsrs} deals with the implications of our results for the skew-metric setting.
We give some background on the remainder evaluation of skew polynomials and the skew metric
in~\autoref{subsec:preliminaries-on-remainder-evaluation} and introduce \ac{FSRS} codes
in~\autoref{subsec:skew-metric-folded}.
\autoref{subsec:decoding-of-fsrs-codes} shows how the presented decoder for \ac{FLRS} codes in the sum-rank metric can
be applied to \ac{FSRS} codes in the skew metric.

Finally,~\autoref{sec:conclusion} concludes the paper by summarizing our work and giving open problems and directions
for further research.

\section{Preliminaries} \label{sec:preliminaries}

Let $q$ be a prime power and let $\Fq$ be a finite field of order $q$.
For any $m \in \NN^{\ast}$, let $\Fqm \supseteq \Fq$ denote an extension field with $q^m$ elements.
We call $\pe \in \Fqm$ \emph{primitive} in $\Fqm$ if it generates the multiplicative group
$\Fqm^{\ast} \defeq \Fqm \setminus \{0\}$.

An \emph{(integer) composition} of $\len \in \NN^{\ast}$ into $\shots \in \NN^{\ast}$ parts, which is also called an
\emph{$\shots$-composition} for short, is a vector $\lenVec = (\lenShot{1}, \dots, \lenShot{\shots}) \in \NN^{\shots}$ with
$\lenShot{i} > 0$ for all $1 \leq i \leq \shots$ that satisfies $\len = \sum_{i=1}^{\shots} n_i$.
We use the notation $\Fq^{\lenVec} \defeq \Fq^{\lenShot{1}} \times \dots \times \Fq^{\lenShot{\shots}}$ to describe the
space of $\Fq^{\len}$-vectors that are divided into $\shots$ blocks with respect to a given $\shots$-composition
$\lenVec$ of $\len$.
Similarly, we write $\Fq^{\heightVec \times \lenVec} \defeq \Fq^{\heightShot{1} \times \lenShot{1}} \times \dots
\times \Fq^{\heightShot{\shots} \times \lenShot{\shots}}$ for $\shots$-compositions $\heightVec$ of $\height$ and
$\lenVec$ of $\len$.
In the following, we always assume $\heightVec \leq \lenVec$ with respect to the product order on $\NN^{\shots}$, that is
$\heightShot{i} \leq \lenShot{i}$ holds for each $i = 1, \dots, \shots$.
\begin{definition}[Sum-Rank Metric]
    The \emph{sum-rank weight} of a tuple
    $\X = ( \shot{\X}{1}, \dots,\allowbreak \shot{\X}{\shots} ) \in \Fq^{\heightVec \times \lenVec}$ is
    \begin{equation}
        \SumRankWeight(\mat{X}) \defeq \sum_{i=1}^{\shots} \rk_q\left(\mat{X}^{(i)}\right)
    \end{equation}
    and the vector $\t = (t_1, \dots, t_{\shots}) \in \NN^{\shots}$ with $t_i \defeq \rk_q \left( \mat{X}^{(i)} \right)$
    for all $i = 1, \dots, \shots$ is called the \emph{weight decomposition} of $\X$.

    The \emph{sum-rank metric} $\SumRankDist$ is defined as
    \begin{equation}
        \SumRankDist(\X, \Y) \defeq \SumRankWeight(\X - \Y)
    \end{equation}
    for two elements $\X, \Y \in \Fq^{\heightVec \times \lenVec}$.
\end{definition}

A \emph{linear sum-rank-metric code} $\mycode{C}$ is an $\Fq$-linear subspace of the metric space
$(\Fq^{\heightVec \times \lenVec},\allowbreak \SumRankDist)$.
Its \emph{minimum (sum-rank) distance} is
\begin{align}
    \SumRankDist(\mycode{C}) &= \min \{\SumRankDist(\C_1, \C_2): \C_1, \C_2 \in \mycode{C}, \C_1 \neq \C_2 \}
    \\
    &= \min \{ \SumRankWeight(\C) : \C \in \mycode{C}, \C \neq 0 \}.
\end{align}

If $\heightVec = (m, \dots, m)$ for some $m \in \NN^{\ast}$, we sometimes write codewords
as $(m \times \len)$-matrices over $\Fq$ instead of matrix tuples from $\Fq^{\heightVec \times \lenVec}$.
Moreover, a code $\mycode{C}$ in this ambient space has a vector-code representation $\mycode{C}_{vec} \defeq
\{ \extInv_{\vecgamma}(\C): \C \in \mycode{C} \} \subseteq \Fqm^{\len}$ over $\Fqm$.
Here, the map $\extInv_{\vecgamma}$ is the inverse of the \emph{extension map} $\ext_{\vecgamma}$ that extends
a vector $\a \in \Fqm^{\len}$ to a matrix $\A \in \Fq^{m \times \len}$ with respect to a fixed ordered $\Fq$-basis $\vecgamma
= (\gamma_1, \dots, \gamma_m)$ of $\Fqm$.
Namely,
\begin{align}
    \ext_{\vecgamma}: \quad \Fqm^{\len} &\to \Fq^{m \times \len},
    \\
    \a = (a_1, \dots, a_{\len}) &\mapsto \A =
    \begin{pmatrix}
        a_{1 1} & \dots & a_{1 \len} \\
        \vdots & \ddots & \vdots \\
        a_{m 1} & \dots & a_{m \len}
    \end{pmatrix}
    \text{ with } a_i = \sum_{j=1}^{m} a_{j i} \gamma_j \text{ for } 1 \leq i \leq \len.
\end{align}
Note that we omit the index $\vecgamma$ if the particular choice of the basis is irrelevant.

The \emph{Frobenius automorphism} of the field extension $\Fqm / \Fq$ is the map $\frob: \Fqm \to \Fqm$ with $\frob(x) = x^q$ for all
$x \in \Fqm$.
It is $\Fq$-linear, fixes $\Fq$ elementwise, and generates the group of all $\Fq$-linear automorphisms on $\Fqm$ with respect to function composition.
We focus on an arbitrary $\Fq$-linear automorphism $\aut$ on $\Fqm$ in the following.
In particular, $\aut = \theta^{\frobPar}$ holds for a $\frobPar \in \{0, \dots, m-1\}$.
A \emph{$\aut$-derivation} is a map $\der: \Fqm \to \Fqm$ satisfying
\begin{equation}
    \der(a + b) = \der(a) + \der(b)
    \quad \text{and} \quad
    \der(ab) = \der(a)b + \aut(a)\der(b)
    \quad \text{for all } a, b \in \Fqm.
\end{equation}
Since we work over finite fields, any $\aut$-derivation is an \emph{inner derivation} which means that $\der = \derPar (\Id - \aut)$
for a $\derPar \in \Fqm$ and the identity $\Id$ (see~\cite[Prop.~44]{martinez2018skew}).

Two elements $a, b \in \Fqm$ are called \emph{$(\aut, \der)$-conjugate} if there is a $c \in \Fqm^{\ast}$ such that
$\conj{a}{c} \defeq \aut(c) a c^{-1} + \der(c) c^{-1} = b$.
This is an equivalence relation and $\Fqm$ is hence partitioned into \emph{conjugacy classes}
$\set{C}(a) \defeq \left\{ a^c : c \in \Fqm^{\ast} \right\}$ for $a \in \Fqm$ (see~e.g.~\cite{lam1985general,lam1988vandermonde}).
A counting argument shows that there are $q^{\gcd(\frobPar, m)}$ distinct conjugacy classes and all except
$\set{C}(\derPar)$ are called nontrivial.
If $\der = 0$ (i.e., $\derPar = 0$) and $\aut = \theta$, the powers $1, \pe, \dots, \pe^{q-2}$ of a primitive element $\pe \in \Fqm^{\ast}$
are representatives of all $q^{\gcd(1,m)}-1 = q-1$ distinct nontrivial conjugacy classes.

The \emph{skew polynomial ring} $\SkewPolyring$ is defined as the set of polynomials $\sum_i f_i x^i$ with finitely many
nonzero coefficients $f_i \in \Fqm$.
It forms a non-commutative ring with respect to ordinary polynomial addition and multiplication determined by the rule
$x f_i = \aut(f_i) x + \der(f_i)$ for all $f_i \in \Fqm$.
We define the \emph{degree} of a skew polynomial $f(x) = \sum_i f_i x^i$ as $\deg(f) \defeq \max \{i : f_i \neq 0\}$ and
write $\SkewPolyring_{<k} \defeq \{ f \in \SkewPolyring: \deg(f) < k \}$ for the set of skew polynomials
of degree less than $k \geq 0$.

We further introduce the operator $\op{a}{b} \defeq \aut(b) a + \der(b)$ for any $a, b \in \Fqm$ and we write
$\opexp{a}{b}{i} \defeq \op{a}{\opexp{a}{b}{i-1}}$ for its $i$-th power with $i \in \NN^{\ast}$.
Let $\lenVec = (\lenShot{1}, \dots, \lenShot{\shots}) \in \NN^{\shots}$ be an $\shots$-composition of $\len \in \NN^{\ast}$.
For a vector $\x = \left(\x^{(1)}, \dots, \x^{(\shots)}\right) \in \Fqm^{\lenVec}$, a vector
$\a = (a_1, \ldots, a_\shots) \in \Fqm^{\shots}$, and a parameter $d \in \NN^{\ast}$ the \emph{generalized Moore matrix}
$\opMoore{d}{\x}{\a}$ is defined as
\begin{align}\label{eq:def_gen_moore_mat}
    \opMoore{d}{\x}{\a} &\defeq
    \left( \opVandermonde{d}{\x^{(1)}}{a_1}, \dots, \opVandermonde{d}{\x^{(\shots)}}{a_\shots} \right)
    \in \Fqm^{\d \times \lenVec},
    \\
    \text{where }
    \opVandermonde{d}{\x^{(i)}}{a_i} &\defeq
    \begin{pmatrix}
        x^{(i)}_1 & \cdots & x^{(i)}_{\lenShot{i}}
        \\
        \op{a_i}{x^{(i)}_1} & \cdots & \op{a_i}{x^{(i)}_{\lenShot{i}}}
        \\[-4pt]
        \vdots & \ddots & \vdots
        \\
        \opexp{a_i}{x^{(i)}_1}{d-1} & \cdots & \opexp{a_i}{x^{(i)}_{\lenShot{i}}}{d-1}
    \end{pmatrix}
    \quad \text{for } 1 \leq i \leq \shots
    \notag
\end{align}
and $\d \defeq (d, \dots, d) \in \NN^{\shots}$.
If $\a$ contains representatives of pairwise distinct nontrivial conjugacy classes of $\Fqm$ and $\rk_{q}\left(\x^{(i)}\right) =
\lenShot{i}$ for all $1 \leq i \leq \shots$, we have by~\cite[Thm.~2]{martinez2018skew}
and~\cite[Thm~4.5]{lam1988vandermonde} that $\rk_{q^m}\left(\opMoore{d}{\x}{\a}\right) = \min(d, \len)$.

The \emph{generalized operator evaluation} of a skew polynomial $f \in \SkewPolyring$ at $b \in \Fqm$ with respect to
the evaluation parameter $a \in \Fqm$ is defined as $\opev{f}{b}{a} = \sum_{i} f_i \opexp{a}{b}{i}$ and can be written in vector-matrix form
using the generalized Moore matrix from~\eqref{eq:def_gen_moore_mat}.
For $\a = (a_1, \dots, a_{\shots}) \in \Fqm^\shots$ and $\x = (\shot{\x}{1}, \dots,
\shot{\x}{\shots}) \in \Fqm^\lenVec$ we use the shorthand notation
\begin{equation}
    \opev{f}{\x}{\a} \defeq (\opev{f}{\shot{\x}{1}}{a_1}, \dots, \opev{f}{\shot{\x}{\shots}}{a_\shots}) \in \Fqm^\lenVec.
\end{equation}
Let $a_1, \dots, a_{\shots}$ be representatives of distinct nontrivial conjugacy classes of $\Fqm$ and consider
$\lenShot{i}$ $\Fq$-linearly independent elements $\zeta_1^{(i)}, \dots, \zeta_{\lenShot{i}}^{(i)} \in \Fqm$ for each
$i = 1, \dots, \shots$.
Then any nonzero $f \in \SkewPolyring$ satisfying $\opev{f}{\zeta_j^{(i)}}{a_i} = 0$ for all $1 \leq j \leq \lenShot{i}$
and all $1 \leq i \leq \shots$ has degree at least $\sum_{i=1}^{\shots} \lenShot{i}$ (see
e.g.~\cite{caruso2019residues}).
\begin{definition}[\Acl{LRS} Codes~{\cite[Def.~31]{martinez2018skew}}]
    \label{def:lrs-codes}
	Let $\a = (a_1, \dots,\allowbreak a_\shots) \allowbreak \in \Fqm^{\shots}$ contain representatives of pairwise distinct
	nontrivial conjugacy classes of $\Fqm$ and consider an $\shots$-composition
    $\n \defeq (n_1, \dots, n_\shots) \in \NN^{\shots}$ of $n \in \NN$.
    Let the vectors $\vecbeta^{(i)} = (\beta_1^{(i)}, \dots, \beta_{n_i}^{(i)}) \in \Fqm^{n_i}$ contain $\Fq$-linearly
    independent elements for all $i = 1, \dots, \shots$ and define
    $\vecbeta \defeq (\vecbeta^{(1)}, \dots, \vecbeta^{(\shots)}) \in \Fqm^{\lenVec}$.
    A \emph{\acl{LRS} code} of length $n$ and dimension
    $k \leq n$ is defined as
    \begin{equation}
        \linRS{\vecbeta,\a;\n,k} \defeq \left\{
        \opev{f}{\vecbeta}{\a}
        : f \in \SkewPolyring_{<k} \right\} \subseteq \Fqm^{\lenVec}.
    \end{equation}
\end{definition}

Observe that the parameter restrictions in~\autoref{def:lrs-codes} also imply restrictions on the length that \ac{LRS} codes can achieve.
Since the evaluation parameters $a_1, \dots, a_{\shots}$ have to belong to distinct nontrivial conjugacy classes, the number of blocks $\shots$
is upper bounded by the number of these classes.
As we know from~\autoref{sec:preliminaries}, $\Fqm$ has $q^{\gcd(u,m)}-1$ distinct nontrivial conjugacy classes, where
$u \in \{0, \dots, m-1\}$ is
defined by the equality $\aut = \theta^u$ for the Frobenius automorphism $\theta$ of $\Fqm/\Fq$.
Thus, $\shots \leq q^{\gcd(u,m)}-1$ has to apply.
At the same time, the code locators $\shot{\vecbeta}{i}$ of the $i$-th block have to contain $\Fq$-linearly
independent elements for all $i = 1, \dots, \shots$ which implies $\lenShot{i} \leq m$.
This means that the length $n$ of an \ac{LRS} code is always bounded by $n \leq (q^{\gcd(u,m)}-1) \cdot m$.

The next lemma is taken from~\cite[Lemma~III.12]{byrne2021fundamental} and lays the foundation for a Singleton-like
bound for sum-rank-metric codes with different block sizes.

\begin{lemma} \label{lem:maximize_elementwise_product}
    Consider an $\shots$-composition $\n = (\lenShot{1}, \dots, \lenShot{\shots})$ of $\len \in \NN$ and a vector
    $\heightVec = (\heightShot{1}, \dots, \heightShot{\shots}) \in \NN^{\shots}$ with $\heightShot{1} \geq \dots \geq \heightShot{\shots} > 0$ and
    $\lenVec \leq \heightVec$.
    Define the set
    \begin{equation}
        \set{U}_z \defeq \left\{ \z = (z_1, \dots, z_{\shots}) \in \NN^{\shots} : \z \leq \lenVec \text{ and }
        \sum_{i=1}^{\shots} z_i = z \right\}
    \end{equation}
    for each $z \in \{0, \dots, \len\}$.
    If we denote by $\ispecial \in \{1, \dots, \shots\}$ and $\lambda \in \{0, \dots, \lenShot{\ispecial} - 1\}$ the unique
    integers that satisfy $\sum_{i=1}^{\ispecial - 1} \lenShot{i} + \lambda = z$,
    then it holds
    \begin{equation}
        \max \left\{ \sum_{i=1}^{\shots} \heightShot{i} z_i : (z_1, \dots, z_{\shots}) \in \set{U}_z \right\}
        = \sum_{i=1}^{\ispecial - 1} \heightShot{i} \lenShot{i} + \heightShot{\ispecial} \lambda.
    \end{equation}
\end{lemma}

We can think about this result in the context of a matrix tuple from $\Fq^{\m \times \n}$ where we are allowed to mark
$z$ columns.
Our goal is then to maximize the number of marked entries which is given as $\sum_{i=1}^{\shots} \heightShot{i} z_i$.
Since the matrices are sorted descendingly with respect to their number of rows, the logical strategy is to mark the
first $z$ columns.
The index $\ispecial$ then corresponds to the first block for which we cannot mark every column anymore.

\begin{theorem}[Singleton-like Bound~{\cite[Thm.~III.2]{byrne2021fundamental}}] \label{thm:singleton_bound}
    Let $\heightVec = (\heightShot{1}, \dots, \heightShot{\shots})$ and $\lenVec = (\lenShot{1}, \dots,
    \lenShot{\shots})$ be integer vectors with $\heightShot{1} \geq \dots \geq \heightShot{\shots} > 0$ and
    $0 < \lenShot{i} \leq \heightShot{i}$ for all $i \in \{1, \dots, \shots\}$.
    Consider a sum-rank-metric code $\mycode{C} \subseteq \Fq^{\heightVec \times \lenVec}$ with
    $\vert \mycode{C} \vert \geq 2$ and $\SumRankDist(\mycode{C}) = d$.
    Then,
    \begin{equation} \label{eq:singleton_bound}
        \vert \mycode{C} \vert
        \leq q^{\sum_{i = \ispecial}^{\shots} \heightShot{i} \lenShot{i} - \heightShot{\ispecial} \lambda},
    \end{equation}
    where $\ispecial \in \{1, \dots, \shots\}$ and $0 \leq \lambda < \lenShot{\ispecial}$ are the unique integers
    such that $d - 1 = \sum_{i=1}^{\ispecial - 1} \lenShot{i} + \lambda$ holds.
\end{theorem}

Note that~\autoref{thm:singleton_bound} generalizes the statements~\cite[Prop.~34]{martinez2018skew} and~\cite[Cor.~2]{martinez2019universal}
that were derived for codes with $\heightShot{1} = \dots = \heightShot{\shots}$.

\section{Folded Linearized Reed--Solomon Codes} \label{sec:flrs-codes}

Code classes obtained by a folding construction have been considered starting from \ac{RS} and Gabidulin codes
in~\cite{Guruswami2008Explicit} and~\cite{BartzSidorenko_FoldedGabidulin2015_DCC,bartz2017algebraic}, respectively.
Let us describe the folding process for a codeword $\c$ of length $n$ and a folding parameter $h$ that divides $n$.
We obtain the folded codeword by subdividing $\c$ into $\frac{n}{h}$ pieces of length $h$ and using them as columns of a
matrix of size $h \times \frac{n}{h}$.
The folded code is simply the collection of all folded codewords.

The folding operation can be expressed by means of the \emph{folding operator}
\begin{align}
    \foldOp{\foldPar}: \qquad \qquad \qquad \Fqm^{\len} &\to \Fqm^{\foldPar \times \lenFLRS} \\
    (x_1, \dots, x_{\len}) &\mapsto
    \begin{pmatrix}
        x_1 & x_{\foldPar + 1} & \dots & x_{\len - \foldPar + 1} \\
        x_2 & x_{\foldPar + 2} & \dots & x_{\len - \foldPar + 2} \\
        \vdots & \vdots & \ddots & \vdots \\
        x_{\foldPar} & x_{2 \foldPar} & \dots & x_{\len}
    \end{pmatrix}
\end{align}
where $\len, \lenFLRS \in \NN^{\ast}$ denote the length of the unfolded and folded vector, respectively, and where the
folding parameter $\foldPar \in \NN^{\ast}$ divides $\len$ with $\lenFLRS = \frac{\len}{\foldPar}$.
Its inverse allows to \emph{unfold} a matrix and is denoted by $\foldOpInv{\foldPar}$.

This paper focuses on folding \ac{LRS} codes which are a generalization of both \ac{RS} and Gabidulin codes.
Since \ac{LRS} codes are naturally equipped with a block structure, we apply the described folding mechanism blockwise
to obtain \ac{FLRS} codes.
Observe that since the length of the blocks may vary, we may choose a different folding parameter
for each block.
This produces sum-rank-metric codes whose codeword tuples consist of matrices with different numbers of rows and columns.
A visual representation of the folding construction for a particular block of an \ac{LRS} codeword is given
in~\autoref{fig:folding}.
A formal description is the following generalization of the above discussed folding operator:
\begin{align}
    \foldOp{\foldParVec}: \qquad \qquad \qquad \Fqm^{\lenVec} &\to \Fqm^{\foldParVec \times \lenFLRSVec} \\
    \left( \shot{\x}{1}, \dots, \shot{\x}{\shots} \right) &\mapsto
    \left( \foldOp{\foldParShot{1}}(\shot{\x}{1}), \dots, \foldOp{\foldParShot{\shots}}(\shot{\x}{\shots}) \right).
\end{align}
Here, vectors of length $\len$ are divided into $\shots$ blocks according to the $\shots$-composition $\lenVec$ and the
vector $\foldParVec = (\foldParShot{1}, \dots, \foldParShot{\shots})$ contains the different folding parameters for the
blocks.
The corresponding inverse map, i.e. the blockwise \emph{unfolding operation}, is denoted by $\foldOpInv{\foldParVec}$.

\begin{definition}[Folded Linearized Reed--Solomon Codes]
    \label{def:flrs-codes}
    Consider an \ac{LRS} code $\mycode{C} \defeq \linRS{\vecbeta,\a;\n,k}$ with $\shot{\vecbeta}{i} \defeq
    (1, \pe, \dots, \pe^{\lenShot{i}-1}) \in \Fqm^{\lenShot{i}}$ for a primitive element
    $\pe$ of $\Fqm$ and all $i = 1, \dots, \shots$.
    Choose a vector $\foldParVec = (\foldParShot{1}, \dots, \foldParShot{\shots}) \in \NN^{\shots}$ of folding
    parameters satisfying $\foldParShot{i} \divides \lenShot{i}$ and $\lenFLRSshot{i} \defeq \frac{\lenShot{i}}{\foldParShot{i}} \leq \foldParShot{i}$ for all $1 \leq i \leq \shots$ and write
    $\N \defeq (\lenFLRSshot{1}, \dots, \lenFLRSshot{\shots})$.
    The $\foldParVec$-folded variant of $\mycode{C}$ is the \emph{$\foldParVec$-folded \acl{LRS} code}
    $\foldedLinRS{\pe, \a, \foldParVec; \lenFLRSVec, k}$ of length $\lenFLRS \defeq \sum_{i=1}^{\shots} \lenFLRSshot{i}$
    and dimension $k$ defined as
    \begin{equation}
        \left\{
        \foldOp{\foldParVec}( \opev{f}{\vecbeta}{\a} )
        = \left(
        \foldOp{\foldParShot{1}}( \opev{f}{\vecbeta^{(1)}}{a_1} ), \dots,
        \foldOp{\foldParShot{\shots}}( \opev{f}{\vecbeta^{(\shots)}}{a_\shots} )
        \right) : f \in \SkewPolyring_{<k} \right\}.
    \end{equation}
\end{definition}

The ambient space of this code is $\Fqm^{\foldParVec \times \lenFLRSVec}$
and we can interpret the folded code as vector code of length $\lenFLRS$ over the field $\Fqd$ with extension degree
$d \defeq m \cdot \lcm (\foldParShot{1}, \dots, \foldParShot{\shots})$ over $\Fq$.
However, linearity is only guaranteed with respect to the subfield $\Fqm$ and due to the $\Fqm$-linearity of the unfolded
\ac{LRS} code.

To make the above definition more explicit, note that there is a message polynomial $f \in \SkewPolyring_{<k}$ for
every codeword $\C \in \foldedLinRS{\pe, \a, \foldParVec; \lenFLRSVec, k} \subseteq
\Fqm^{\foldParVec \times \lenFLRSVec}$ with
\begin{equation}
    \C = \left( \C^{(1)}(f), \dots, \C^{(\shots)}(f) \right)
\end{equation}
\begin{equation}\label{eq:defFLRScodeblock}
    \text{and }
    \C^{(i)}(f) \defeq
    \begin{pmatrix}
        \opev{f}{1}{a_i} & \opev{f}{\pe^{\foldParShot{i}}}{a_i} & \cdots
        & \opev{f}{\pe^{\lenShot{i}-\foldParShot{i}}}{a_i}
        \\
        \opev{f}{\pe}{a_i} & \opev{f}{\pe^{\foldParShot{i}+1}}{a_i} & \cdots
        & \opev{f}{\pe^{\lenShot{i}-\foldParShot{i}+1}}{a_i}
        \\
        \vdots & \vdots & \ddots & \vdots
        \\
        \opev{f}{\pe^{\foldParShot{i}-1}}{a_i} & \opev{f}{\pe^{2\foldParShot{i}-1}}{a_i} & \cdots
        & \opev{f}{\pe^{\lenShot{i}-1}}{a_i}
    \end{pmatrix}
    \in \Fqm^{\foldParShot{i} \times \lenFLRSshot{i}}
\end{equation}
for all $i \in \{1, \ldots, \shots\}$.

We can further draw conclusions about the maximum length of \ac{FLRS} codes, similar to the \ac{LRS} case.
Let us therefore assume that the parameters of the unfolded code are maximal.
In other words, choose an \ac{LRS} code $\mycode{C}$ in~\autoref{def:flrs-codes} with $\shots = q^{\gcd(u,m)} - 1$
same-sized blocks of length $m$ and resulting overall code length $\len = (q^{\gcd(u,m)} - 1) \cdot m$ (see~\autoref{sec:preliminaries} for the definition of $u$ and the derivation of this statement).
Since we want to maximize the length of the folded code, we choose $\foldParShot{i}$ for each $i = 1, \dots, \shots$ as small as possible such that
$\foldParShot{i} \divides \lenShot{i}$ and $\lenFLRSshot{i} \defeq \frac{\lenShot{i}}{\foldParShot{i}} \leq \foldParShot{i}$ hold.
As all blocks have the same length, we select the same folding parameter $h$ for each block and it has to satisfy
$h \divides m$ and $m \leq h^2$.
We cannot get any better than $h = \sqrt{m}$ and thus obtain the upper bound $\lenFLRS \leq (q^{\gcd(u,m)} - 1) \cdot \sqrt{m}$
on the total length $\lenFLRS$ of \ac{FLRS} codes.

\begin{remark}
    Note that we only consider a subclass of \ac{LRS} codes for folding.
	Namely, we choose the code locators as powers of a primitive element $\pe \in \Fqm^{\ast}$.
	This turns out to be crucial for the interpolation-based decoder that we present in~\autoref{sec:decoding}.
\end{remark}

\begin{figure}
    \centering
    \begin{subfigure}[b]{.48\textwidth}
        \centering
        \resizebox{\textwidth}{!}{%
\begin{tikzpicture}
	\foreach \x [count=\i] in {0,...,11} {
		\node (node\x) at (.5*\x, 0) {$c_{\i}$};
	}

	\foreach \b in {1,4,7,10} {
		\node (block\b) at (.5*\b, 0) [draw, minimum width=1.5cm, minimum height=.5cm] {};
	}
\end{tikzpicture}
        }
        \caption{Codeword block $\shot{\c}{i} = (c_1, \dots, c_{12})$ cut into blocks of length $h_i=3$.}
        \label{fig:folding-2}
    \end{subfigure}
    \vspace*{.3cm}
    \hfill
    \begin{subfigure}[b]{.48\textwidth}
        \centering
        \resizebox{\textwidth}{!}{%
\begin{tikzpicture}
	\foreach \y [count=\i] in {0,1,2} {
		\node (node\y) at (.5, -\y/2+.5) {$c_{\i}$};
	}
	\foreach \y [evaluate=\y as \i using int(\y+1)] in {3,4,5} {
		\node (node\y) at (2.5, -\y/2+2) {$c_{\i}$};
	}
	\foreach \y [evaluate=\y as \i using int(\y+1)] in {6,7,8} {
		\node (node\y) at (4.5, -\y/2+3.5) {$c_{\i}$};
	}
	\foreach \y [evaluate=\y as \i using int(\y+1)] in {9,10,11} {
		\node (node\y) at (6.5, -\y/2+5) {$c_{\i}$};
	}

	\foreach \b in {.5,2.5,4.5,6.5} {
		\node (block\b) at (\b, 0) [draw, minimum width=.5cm, minimum height=1.5cm] {};
	}

	\draw [blue] (.5, 1) -- (.5, -1) -- (2.5, 1) -- (2.5, -1) -- (4.5, 1) -- (4.5, -1) -- (6.5, 1) -- (6.5, -1);
\end{tikzpicture}
        }
        \caption{$3$-folded version of $\shot{\c}{i}$. The blue line illustrates the terminology ``folding''.}
        \label{fig:folding-5}
    \end{subfigure}
    \caption{
        Illustration of the folding construction for a block $\c^{(i)} = (c_1, \ldots, c_{12})$ of an \ac{LRS}
        codeword $\c = (\shot{\c}{1}, \dots, \shot{\c}{\shots})$ using folding parameter $h_i=3$.
    }
    \label{fig:folding}
\end{figure}
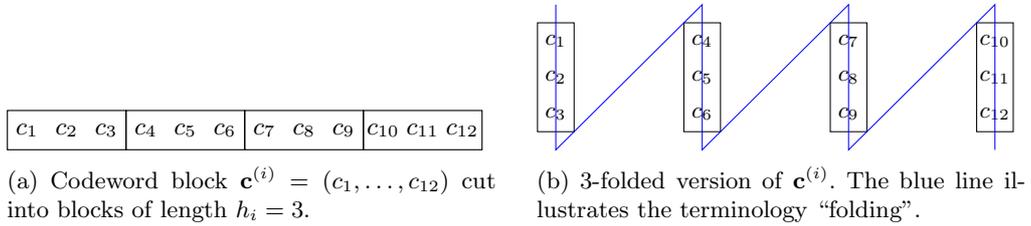

\begin{theorem}[Minimum Distance]\label{thm:minimum_distance_flrs}
    Let $\mycode{C} \defeq \foldedLinRS{\pe, \a, \foldParVec; \lenFLRSVec, k}$ be an \ac{FLRS} code and assume without
    loss of generality that $\foldParShot{1} \geq \dots \geq \foldParShot{\shots}$ applies.
    The minimum sum-rank distance of $\mycode{C}$ is
    \begin{equation}
        \SumRankDist(\mycode{C}) = \sum_{i=1}^{\ispecial} \lenFLRSshot{i}
        - \left\lceil \frac{k - \sum_{i=\ispecial + 1}^{\shots} \foldParShot{i} \lenFLRSshot{i} - 1}
        {\foldParShot{\ispecial}} \right\rceil + 1,
    \end{equation}
    where $\ispecial \in \{1, \dots, \shots\}$ is the unique choice that satisfies
    \begin{equation}
        0 \leq \SumRankDist(\mycode{C}) - \sum_{i=1}^{\ispecial - 1} \lenFLRSshot{i} - 1 <
        \lenFLRSshot{\ispecial}.
    \end{equation}
    In particular, $\mycode{C}$ achieves the Singleton-like bound~\eqref{eq:singleton_bound} with equality if and only if
    $\foldParShot{\ispecial}$ divides both $k$ and $\foldParShot{i} \lenFLRSshot{i}$ for all
    $i = \ispecial + 1, \dots, \shots$.
\end{theorem}

\begin{proof}
    Let $\C = (\shot{\C}{1}, \dots, \shot{\C}{\shots})\in \mycode{C}$ be the nonzero codeword corresponding to the message polynomial $f \in \SkewPolyring_{<k}$.
    Then there are $z, z_1, \dots, z_{\shots} \geq 0$ with $z = \sum_{i=1}^{\shots} z_i$ such that
    $\SumRankWeight(\C) = \lenFLRS - z$ and $\rk_q(\C^{(i)}) = \lenFLRSshot{i} - z_i$ for $i = 1, \dots, \shots$.
    Let us denote by $\RCEF(\C) \in \Fqm^{\foldParVec \times \lenFLRSVec}$ the blockwise-reduced column-echelon form of
    $\C$ which is obtained by bringing each block $\shot{\C}{i}$ independently in its reduced column-echelon form with
    respect to $\Fq$ as follows:
    first obtain a matrix $\shot{\C}{i}_q \in \Fq^{m \foldParShot{i} \times \lenShot{i}}$ by replacing each row of the
    block $\shot{\C}{i}$ with its extended matrix that is obtained via $\ext_{\vecgamma}$ for an arbitrary
    $\Fq$-basis $\vecgamma$ of $\Fqm$.
    Next, bring $\shot{\C}{i}_q$ in reduced column-echelon form (e.g. by Gaussian elimination) and finally apply the
    inverse operation $\extInv_{\vecgamma}$ to the matrix blocks and get back a
    $(\foldParShot{i} \times \lenShot{i})$-matrix over $\Fqm$.
    Since $\extInv_{\vecgamma}$ preserves the zero columns and $\rk_q(\C^{(i)}) = \lenFLRSshot{i} - z_i$,
    the number of nonzero columns in the $i$-th block of $\RCEF(\C)$ is $\lenFLRSshot{i} - z_i$.
    Thus, the overall number of nonzero entries is certainly upper-bounded by
    $\sum_{i=1}^{\shots} \foldParShot{i} (\lenFLRSshot{i} - z_i)$.
    We further obtain an upper bound on the last result by finding the vector $\z = (z_1, \dots, z_{\shots})$ that realizes
    \begin{equation}
        \max \left\{ \sum_{i=1}^{\shots} \foldParShot{i} (\lenFLRSshot{i} - z_i) : \lenFLRSVec - \z \leq \lenFLRSVec
        \text{ and } \lenFLRS - z = \sum_{i=1}^{\shots} z_i \right\}
    \end{equation}
    for a fixed $z$.
    With~\autoref{lem:maximize_elementwise_product}, the maximum equals
    $\sum_{i=1}^{\ispecialTwo - 1} \foldParShot{i} \lenFLRSshot{i} + \foldParShot{\ispecialTwo} \varepsilon$,
    where $\ispecialTwo \in \{1, \dots, \shots\}$ and $0 \leq \varepsilon < \lenFLRSshot{\ispecialTwo}$ are the unique
    integers such that $\lenFLRS - z = \sum_{i=1}^{\ispecialTwo - 1} \lenFLRSshot{i} + \varepsilon$.
    When we shift the focus to the zero entries of $\RCEF(\C)$, we naturally obtain the lower bound
    $\sum_{i=\ispecialTwo}^{\shots} \foldParShot{i} \lenFLRSshot{i} - \foldParShot{\ispecialTwo} \varepsilon$ with $y$
    and $\varepsilon$ as before,
    since the number of zero and nonzero entries adds up to $\sum_{i=1}^{\shots} \foldParShot{i} \lenFLRSshot{i}$.

    Note that the $\Fqm$-semilinearity of the generalized operator evaluation ensures that for each $1 \leq i \leq
    \shots$ the entries of the $i$-th block of $\RCEF(\C)$ are still $\Fq$-linearly independent and can be expressed as
    evaluations of $f$ with respect to the evaluation parameter $a_i$.
    Since the number of $\Fq$-linearly independent roots of $f$ with respect to evaluation parameters from distinct
    nontrivial conjugacy classes is bounded by its degree and $\deg(f) \leq k - 1$, we get
    \begin{equation} \label{eq:min_dist_proof_lower_bound}
        \sum_{i=\ispecialTwo}^{\shots} \foldParShot{i} \lenFLRSshot{i}
        - \foldParShot{\ispecialTwo} \varepsilon \leq k - 1
        \quad
        \begin{gathered}
            \text{for } \ispecialTwo \in \{1, \dots, \shots\}, \varepsilon \in \{0, \dots, \lenFLRSshot{\ispecialTwo} - 1\} \\
            \text{unique with } \lenFLRS - z = \sum\nolimits_{i=1}^{\ispecialTwo - 1} \lenFLRSshot{i} + \varepsilon.
        \end{gathered}
    \end{equation}
    On the other hand, the Singleton-like bound for $\Fq$-linear sum-rank-metric codes (see \autoref{thm:singleton_bound}) yields
    \begin{equation} \label{eq:min_dist_proof_upper_bound}
        k \leq \sum_{i = \ispecial}^{\shots} \foldParShot{i} \lenFLRSshot{i}
        - \foldParShot{\ispecial} \lambda
        \quad
        \begin{gathered}
            \text{for } \ispecial \in \{1, \dots, \shots\}, \lambda \in
            \{0, \dots, \lenFLRSshot{\ispecial} - 1\} \\
            \text{unique with } \SumRankDist(\mycode{C}) - 1 =
            \sum\nolimits_{i=1}^{\ispecial - 1} \lenFLRSshot{i} + \lambda.
        \end{gathered}
    \end{equation}
    As we can choose a minimum-weight codeword $\C$ in the above reasoning, we can replace $\lenFLRS - z$ by
    $\SumRankDist(\mycode{C})$ in~\eqref{eq:min_dist_proof_lower_bound}.
    But then there are only two possibilities for the relationship between the indices $\ispecialTwo$ and $\ispecial$
    and the parameters $\varepsilon$ and $\lambda$.
    Namely,
    \begin{enumerate}
        \item $\ispecial = \ispecialTwo$, \tabto{2.5cm} $\varepsilon \in \{1, \dots, \lenFLRSshot{\ispecialTwo} - 1\}$,
                \tabto{6.5cm} and \tabto{8cm} $\lambda = \varepsilon - 1$ \tabto{11cm} or
        \item $\ispecial = \ispecialTwo - 1$, \tabto{2.5cm} $\varepsilon = 0$, \tabto{6.5cm} and
                \tabto{8cm} $\lambda = \lenFLRSshot{\ispecialTwo - 1} - 1$.
    \end{enumerate}
    Let us focus on the first case.
    We get
    \begin{gather}
        \label{eq:min_dist_proof_transformation_start_lower_bound}
        \sum_{i=\ispecialTwo}^{\shots} \foldParShot{i} \lenFLRSshot{i}
        - \foldParShot{\ispecialTwo} \left( \SumRankDist(\mycode{C})
        - \sum_{i=1}^{\ispecialTwo - 1} \lenFLRSshot{i} \right) + 1 \leq k
    \end{gather}
    by substituting $\varepsilon$ for the equality condition in~\eqref{eq:min_dist_proof_lower_bound}.
    We then shift the first summand $\foldParShot{\ispecialTwo} \lenFLRSshot{\ispecialTwo}$ of the first sum into
    the second sum, do some transformations, and finally use the integrality of the left-hand side as well as the fact
    $\lfloor x \rfloor = \lceil x - 1 \rceil$ for any real number $x$ to obtain
    \begin{equation} \label{eq:min_dist_proof_transformed_lower_bound}
        \sum_{i=1}^{\ispecialTwo} \lenFLRSshot{i} - \SumRankDist(\mycode{C}) + 1
        \leq \left\lceil \frac{k - \sum_{i=\ispecialTwo + 1}^{\shots} \foldParShot{i} \lenFLRSshot{i} - 1}
        {\foldParShot{\ispecialTwo}} \right\rceil.
    \end{equation}
    Similarly, substituting $\lambda$ for the equality condition in~\eqref{eq:min_dist_proof_upper_bound} yields
    \begin{gather}
        k \leq
        \sum_{i = \ispecial}^{\shots} \foldParShot{i} \lenFLRSshot{i}
        - \foldParShot{\ispecial} \left( \SumRankDist(\mycode{C}) - 1
        - \sum_{i=1}^{\ispecial - 1} \lenFLRSshot{i} \right)
        \\
        \label{eq:min_dist_proof_transformed_upper_bound}
        \iff
        \left\lceil \frac{k - \sum_{i = \ispecial + 1}^{\shots} \foldParShot{i} \lenFLRSshot{i}}
        {\foldParShot{\ispecial}} \right\rceil
        \leq \sum_{i=1}^{\ispecial} \lenFLRSshot{i} - \SumRankDist(\mycode{C}) + 1.
    \end{gather}
    As $\ispecial = \ispecialTwo$ holds, the right-hand side of~\eqref{eq:min_dist_proof_transformed_lower_bound} is
    less than or equal to the left-hand side of~\eqref{eq:min_dist_proof_transformed_upper_bound}.
    But since the left-hand side of~\eqref{eq:min_dist_proof_transformed_lower_bound} and the right-hand side
    of~\eqref{eq:min_dist_proof_transformed_upper_bound} are equal, all inequalities in the chain must be equalities and
    we get
    \begin{equation} \label{eq:min_dist_proof_result}
        \SumRankDist(\mycode{C}) = \sum_{i=1}^{\ispecial} \lenFLRSshot{i}
        - \left\lceil \frac{k - \sum_{i=\ispecial + 1}^{\shots} \foldParShot{i} \lenFLRSshot{i} - 1}
        {\foldParShot{\ispecial}} \right\rceil + 1,
    \end{equation}
    where $\ispecial \in \{1, \dots, \shots\}$ is unique with $0 \leq
    \SumRankDist(\mycode{C}) - \sum\nolimits_{i=1}^{\ispecial - 1} \lenFLRSshot{i} - 1 <
    \lenFLRSshot{\ispecial}$.

    Let us move on to the second case and recall that $\varepsilon = 0$.
    Therefore, we can replace the factor $\foldParShot{\ispecialTwo}$ of $\varepsilon$ (i.e.\ of
    $\SumRankDist(\mycode{C}) - \sum_{i=1}^{\ispecialTwo - 1} \lenFLRSshot{i}$)
    in~\eqref{eq:min_dist_proof_transformation_start_lower_bound} with $\foldParShot{\ispecialTwo - 1}$.
    Similar transformations as above and again the integrality of the left-hand side yield
    \begin{equation} \label{eq:min_dist_proof_differently_transformed_lower_bound}
        \sum_{i=1}^{\ispecialTwo - 1} \lenFLRSshot{i} - \SumRankDist(\mycode{C}) + 1
        \leq \left\lceil \frac{k - \sum_{i=\ispecialTwo}^{\shots} \foldParShot{i} \lenFLRSshot{i} - 1}
        {\foldParShot{\ispecialTwo - 1}} \right\rceil.
    \end{equation}
    Since we have $\ispecial = \ispecialTwo - 1$ in this case, the right-hand side
    of~\eqref{eq:min_dist_proof_differently_transformed_lower_bound} is less than or equal to the left-hand side
    of~\eqref{eq:min_dist_proof_transformed_upper_bound}.
    But, as in the first case, the left-hand side of~\eqref{eq:min_dist_proof_differently_transformed_lower_bound} and
    the right-hand side of~\eqref{eq:min_dist_proof_transformed_upper_bound} are equal.
    Hence, the inequalities are in fact equalities and we obtain~\eqref{eq:min_dist_proof_result}.

    The Singleton-like bound is met if and only if $\foldParShot{\ispecialTwo}$ divides
    $k - \sum_{i=\ispecialTwo + 1}^{\shots} \foldParShot{i} \lenFLRSshot{i}$, which is equivalent to
    $\foldParShot{\ispecialTwo}$ dividing $k$ as well as $\foldParShot{i} \lenFLRSshot{i}$ for each
    $i = 1, \dots, \shots$.
    This concludes the proof.
\end{proof}

\begin{remark}
    \autoref{thm:minimum_distance_flrs} needs an \ac{FLRS} code to satisfy $\foldParShot{1} \geq \dots \geq
    \foldParShot{\shots}$ for technical reasons.
    However, this is not a restriction since we can simply reorder the blocks of a sum-rank-metric code without changing
    its weight distribution or its minimum distance.
    Formally speaking, we choose a permutation $\pi$ from the symmetric group $\set{S}_{\shots}$ for which
    $\foldParShot{\piInv(1)} \geq \dots \geq \foldParShot{\piInv(\shots)}$ holds and consider
    \begin{equation}
        \tilde{\mycode{C}} = \{ (\shot{\C}{\piInv(1)}, \dots, \shot{\C}{\piInv(\shots)}) :
        (\shot{\C}{1}, \dots, \shot{\C}{\shots}) \in \mycode{C} \}.
    \end{equation}
\end{remark}

\section{Interpolation-Based Decoding of Folded Linearized Reed--Solomon Codes} \label{sec:decoding}

In this section we derive an interpolation-based decoder for \ac{FLRS} codes that is based on the Guruswami--Rudra decoder for \ac{FRS} codes~\cite{Guruswami2008Explicit} and the Mahdavifahr--Vardy decoder for folded Gabidulin codes~\cite{Mahdavifar2012Listdecoding}.
As channel model we consider an additive sum-rank channel that relates the input $\C \in \Fq^{\foldParVec \times \lenFLRSVec}$
to the received output $\R \in \Fqm^{\foldParVec \times \lenFLRSVec}$ by adding an error $\E \in
\Fqm^{\foldParVec \times \lenFLRSVec}$, that is $\R = \C + \E$.
The addition in $\Fqm^{\foldParVec \times \lenFLRSVec}$ is performed componentwise.

We denote the sum-rank weight of the error $\E = (\shot{\E}{1}, \dots, \shot{\E}{\shots}) \in
\Fqm^{\foldParVec \times \lenFLRSVec}$ by $\SumRankWeight(\E) = t$
and its weight decomposition by $\t = (t_1, \dots, t_{\shots})$ with $t_i = \rk_{q} (\shot{\E}{i})$ for all
$i = 1, \dots, \shots$.
$\E$ is chosen uniformly at random from the set of all tuples in $\Fqm^{\foldParVec \times \lenFLRSVec}$ having a fixed
weight $t$ as well as a weight decomposition belonging to a prescribed set of decompositions.

Suppose we transmit a codeword $\C \in \foldedLinRS{\pe, \a, \foldParVec; \lenFLRSVec, k}$ and we receive the word $\R = (\shot{\R}{1}, \dots, \shot{\R}{\shots})$ with
\begin{equation} \label{eq:R_notation}
    \shot{\R}{i} =
    \begin{pmatrix}
        r_1^{(i)} & r_{\foldParShot{i}+1}^{(i)} & \cdots & r_{\lenShot{i}-\foldParShot{i}+1}^{(i)}
        \\
        \vdots & \vdots & \ddots & \vdots
        \\
        r_{\foldParShot{i}}^{(i)} & r_{2\foldParShot{i}}^{(i)} & \cdots & r_{\lenShot{i}}^{(i)}
    \end{pmatrix}
    \in \Fqm^{\foldParShot{i} \times \lenFLRSshot{i}}
    \quad \text{for all } i \in \{1, \ldots, \shots\}.
\end{equation}

Note that the decodability of a specific error will in general depend on its weight decomposition and not only on the
chosen code, the error weight $t$, and the decoder's parameters (see~\autoref{thm:decodingRadius}).
When we consider codes using the same folding parameter $\foldPar \in \NN^{\ast}$ for all blocks, only the error weight
$t$ decides if an error is decodable for the chosen code and decoder.

As is typical for interpolation-based decoding, our decoder consists of two steps that we will describe in the following:
interpolation and root finding.
In the first phase, we construct interpolation points from the received word $\R$ and obtain a multivariate skew
polynomial $Q$ that satisfies certain conditions.
In the second phase, we use the interpolation polynomial $Q$ to find candidates for the message polynomial and hence for
the transmitted codeword.

\subsection{Interpolation Step} \label{subsec:interpolation-step}

We first choose an interpolation parameter $\intDim \in \NN^{\ast}$ satisfying
\begin{equation}\label{eq:intDimConstraint}
    \intDim \leq \min_{i \in \{1, \dots, \shots\}} \foldParShot{i},
\end{equation}
where the constraint arises from the selection method of the interpolation points.
The latter are elements of $\Fqm^{\intDim + 1}$ whose last $\intDim$ entries are obtained from the received word $\R$
using a sliding-window approach.
Namely, we place a window of size $\intDim \times 1$ on the top left corner of $\R$ and slide it down one position at a
time as long as each position of the window covers an entry of $\R$.
Then the window is moved to the next column, starting the same process again from the top.
The first entry of an interpolation point obtained in this way is the code locator corresponding to the window's
starting position, that is a power of the primitive element $\pe$.

Formally speaking, we consider for each $i = 1, \dots, \shots$ the two sets
\begin{gather}\label{eq:intPointDef}
    \begin{aligned}
         \set{W}_i &\defeq \left\{ (j-1) \foldParShot{i} + l :
         j \in \{1, \ldots, \lenFLRSshot{i}\}, l \in \{1, \ldots, \foldParShot{i} - \intDim + 1\} \right\}
         \\
        \text{and} \quad
        \set{P}_i &\defeq \left\{ \left( \pe^{w-1}, r_{w}^{(i)}, r_{w+1}^{(i)}, \dots, r_{w+\intDim-1}^{(i)} \right):
        w \in \set{W}_i \right\},
    \end{aligned}
\end{gather}
where $\set{P}_i$ contains all interpolation points corresponding to the $i$-th block $\shot{\R}{i}$ of $\R$.
The index set $\set{W}_i$ consists of all eligible starting positions for the sliding window within $\shot{\R}{i}$ and
each interpolation point can be naturally identified with a tuple $(w,i)$ with $w \in \set{W}_i$ and $i \in
\{1, \dots, \shots\}$.
Note that, by construction, the set of all interpolation points $\set{P} \defeq \bigcup_{i=1}^{\shots} \set{P}_i$ has
cardinality
\begin{equation}
    \vert \set{P} \vert = \sum_{i=1}^{\shots} \lenFLRSshot{i} (\foldParShot{i} - s + 1).
\end{equation}

\begin{example} \label{ex:int_points}
    Consider an \ac{FLRS} code with folding parameters $\foldParVec = (3, 2)$ and folded block lengths
    $\lenFLRSVec = (2, 2)$.
    Choose $\intDim = 2$ and denote $\R$ according to~\eqref{eq:R_notation}, i.e.\
    \begin{equation}
        \R =
        \left(
        \begin{pmatrix}
            r_1^{(1)} & r_4^{(1)} \\
            r_2^{(1)} & r_5^{(1)} \\
            r_3^{(1)} & r_6^{(1)}
        \end{pmatrix}
        ,
        \begin{pmatrix}
            r_1^{(2)} & r_3^{(2)} \\
            r_2^{(2)} & r_4^{(2)}
        \end{pmatrix}
        \right).
    \end{equation}
    Then the set of interpolation points is the union of
    \begin{align}
        \set{P}_1 &= \Big\{ (1, r_1^{(1)}, r_2^{(1)}), (\pe, r_2^{(1)}, r_3^{(1)}), (\pe^3, r_4^{(1)}, r_5^{(1)}),
        (\pe^4, r_5^{(1)}, r_6^{(1)}) \Big\}
        \\
        \text{and} \quad \set{P}_2 &= \Big\{ (1, r_1^{(2)}, r_2^{(2)}), (\pe^2, r_3^{(2)}, r_4^{(2)}) \Big\}.
    \end{align}
\end{example}

We wish to find a multivariate skew interpolation polynomial $Q$ that satisfies certain interpolation constraints and
has the form
\begin{equation}\label{eq:mult_var_skew_poly_FLRS}
    Q(x, y_1, \dots, y_{\intDim}) = Q_0(x) + Q_1(x) y_1 + \dots + Q_{\intDim}(x) y_{\intDim},
\end{equation}
where $Q_r(x) \in \SkewPolyring$ for all $r \in \{0, \ldots, \intDim\}$.
The \emph{generalized operator evaluation} of such a polynomial $Q \in \MultSkewPolyring$ at a given interpolation point
$\p = (p_0, \dots, p_{\intDim}) \in \Fqm^{\intDim + 1}$ with respect to an evaluation parameter $a \in \Fqm$ is defined
as
\begin{equation}
    \label{eq:defFuncGenOpFLRS}
    \mathscr{E}_{Q}(\p)_{a} \defeq
    \opev{Q_0}{p_0}{a} + \opev{Q_1}{p_1}{a} + \dots + \opev{Q_\intDim}{p_{\intDim}}{a}.
\end{equation}

\begin{problem}[Interpolation Problem] \label{prob:intProblemFLRS}
    Let a parameter $\degConstraint \in \NN^{\ast}$, a set $\set{P} = \bigcup_{i=1}^{\shots} \set{P}_i$
    of interpolation points and evaluation parameters $a_1, \dots, a_{\shots}$ be given.
    Find a nonzero $(s+1)$-variate skew polynomial $Q$ of the
    form~\eqref{eq:mult_var_skew_poly_FLRS} satisfying
    \begin{enumerate}
        \item
            $\mathscr{E}_{Q}(\p)_{a_i} = 0$ for all $\p \in \set{P}_i$ and $i \in \{1, \ldots, \shots\}$ as well as
        \item
            $\deg(Q_0) < D$ and $\deg(Q_r) < D - k + 1$ for all $r \in \{1, \ldots, \intDim\}$.
    \end{enumerate}
\end{problem}
Note that the evaluation parameters $a_1, \dots, a_{\shots}$ are the entries of $\a$ of the considered \ac{FLRS} code
$\foldedLinRS{\pe, \a, \foldParVec; \lenFLRSVec, k}$.

\autoref{prob:intProblemFLRS} can be solved using skew Kötter interpolation from~\cite{liu2014kotter}
(similar as in~\cite[Sec.~V]{bartz2014efficient}) requiring at most $\OCompl{\intDim \len^2}$ operations in $\Fqm$ in the zero-derivation case.
There exist fast interpolation algorithms~\cite{bartz2022fast, bartz2022knh} that can solve~\autoref{prob:intProblemFLRS} requiring at most $\softoh{\intDim^\omega \OMul{\len}}$ operations in $\Fqm$ in the zero-derivation case, where $\softoh{\cdot}$ denotes the \emph{soft}-O notation (which neglects log factors), $\OMul{n}\in\OCompl{n^{1.635}}$ is the cost of multiplying two skew-polynomials of degree at most $n$ and $\omega < 2.37286$ is the matrix multiplication exponent~\cite{le_gall_powers_2014}.

Since the second condition of the interpolation problem allows us to write
\begin{equation}
    \label{eq:Qcoefficients}
    Q_0(x) = \sum_{j=0}^{\degConstraint-1} q_{0, j} x^j
    \quad \text{and} \quad
    Q_r(x) = \sum_{j=0}^{\degConstraint-k} q_{r, j} x^j
    \quad \text{for} \quad
    r \in \{1, \ldots, \intDim\}
\end{equation}
with all coefficients from $\Fqm$, we can also solve~\autoref{prob:intProblemFLRS} by solving a system of $\Fqm$-linear
equations whose coefficient matrix describes the
first condition of the interpolation problem.
We collect all interpolation points from $\set{P}_i$ as rows in a matrix $\P_i \in
\Fqm^{\lenFLRSshot{i}(\foldParShot{i} - \intDim + 1) \times (\intDim + 1)}$ for each $1 \leq i \leq \shots$ and denote
its columns by $\p_{i, 0}, \ldots, \p_{i, \intDim}$.
Define further $\p_r = (\p_{1, r}^{\top} \ | \ \cdots \ | \ \p_{\shots, r}^{\top})$ for $0 \leq r \leq \intDim$.
Then,~\autoref{prob:intProblemFLRS} can be written as
\begin{gather}\label{eq:intSystem}
    \S\q_I^{\top} = \0 \\
    \text{with} \quad
    \S = \left(
    \begin{array}{c|c|c|c}
        \left(\opMoore{D}{\p_0}{\a}\right)^{\top} &
        \left(\opMoore{D-k+1}{\p_1}{\a}\right)^{\top} &
        \cdots &
        \left(\opMoore{D-k+1}{\p_s}{\a}\right)^{\top}
    \end{array}
    \right)
    \notag \\
    \text{and} \quad
    \q_I = \left(
    q_{0, 0} \cdots q_{0, D-1} \ | \ q_{1, 0} \cdots q_{1, D-k} \ | \ \cdots \ | \
    q_{\intDim, 0} \cdots q_{\intDim, D-k}
    \right).
    \notag
\end{gather}

\begin{example} \label{ex:int_matrix}
    Let us continue~\autoref{ex:int_points} with $k = 2$ and derive the corresponding interpolation matrix for the choice
    $\degConstraint = 3$.
    As we will see shortly in~\autoref{lem:degConstraintForExistence}, this choice guarantees the existence of a nonzero
    solution of~\eqref{eq:intSystem}.
    We get
    \begin{equation}
        \P_1 =
        \begin{pmatrix}
            1 & r_1^{(1)} & r_2^{(1)} \\
            \pe & r_2^{(1)} & r_3^{(1)} \\
            \pe^3 & r_4^{(1)} & r_5^{(1)} \\
            \pe^4 & r_5^{(1)} & r_6^{(1)}
        \end{pmatrix}
        \qquad \text{ and } \qquad
        \P_2 =
        \begin{pmatrix}
            1 & r_1^{(2)} & r_2^{(2)} \\
            \pe^2 & r_3^{(2)} & r_4^{(2)} \\
        \end{pmatrix}
    \end{equation}
    and hence the interpolation matrix $\S$ is given by
    \begin{gather}
        \left(
        \begin{array}{c|c|c}
            \opMoore{3}{\p_0}{\a}^{\top}
            & \opMoore{2}{\p_1}{\a}^{\top}
            & \opMoore{2}{\p_2}{\a}^{\top}
        \end{array}
        \right)
        =
        \left(
        \begin{array}{c|c|c}
            \opVandermonde{3}{\p_{1, 0}}{a_1}^{\top}
            & \opVandermonde{2}{\p_{1, 1}}{a_1}^{\top}
            & \opVandermonde{2}{\p_{1, 2}}{a_1}^{\top}
            \\
            \opVandermonde{3}{\p_{2, 0}}{a_2}^{\top}
            & \opVandermonde{2}{\p_{2, 1}}{a_2}^{\top}
            & \opVandermonde{2}{\p_{2, 2}}{a_2}^{\top}
        \end{array}
        \right),
        \\
        \text{i.e.\ } \quad
        \S = \left(
        \begin{array}{ccc|cc|cc}
            1 & \op{a_1}{1} & \opexp{a_1}{1}{2} &
            r_1^{(1)} & \op{a_1}{r_1^{(1)}} &
            r_2^{(1)} & \op{a_1}{r_2^{(1)}}
            \\
            \pe & \op{a_1}{\pe} & \opexp{a_1}{\pe}{2} &
            r_2^{(1)} & \op{a_1}{r_2^{(1)}} &
            r_3^{(1)} & \op{a_1}{r_3^{(1)}}
            \\
            \pe^3 & \op{a_1}{\pe^3} & \opexp{a_1}{\pe^3}{2} &
            r_4^{(1)} & \op{a_1}{r_4^{(1)}} &
            r_5^{(1)} & \op{a_1}{r_5^{(1)}}
            \\
            \pe^4 & \op{a_1}{\pe^4} & \opexp{a_1}{\pe^4}{2} &
            r_5^{(1)} & \op{a_1}{r_5^{(1)}} &
            r_6^{(1)} & \op{a_1}{r_6^{(1)}}
            \\
            \hline
            1 & \op{a_2}{1} & \opexp{a_2}{1}{2} &
            r_1^{(2)} & \op{a_2}{r_1^{(2)}} &
            r_2^{(2)} & \op{a_2}{r_2^{(2)}}
            \\
            \pe^2 & \op{a_2}{\pe^2} & \opexp{a_2}{\pe^2}{2} &
            r_3^{(2)} & \op{a_2}{r_3^{(2)}} &
            r_4^{(2)} & \op{a_2}{r_4^{(2)}}
        \end{array}
        \right).
    \end{gather}
\end{example}

\begin{lemma}[Existence]\label{lem:degConstraintForExistence}
    A nonzero solution to~\autoref{prob:intProblemFLRS} exists if
    \begin{equation} \label{eq:deg_constraint}
        \degConstraint = \left\lceil
        \frac{\sum_{i=1}^{\shots} \lenFLRSshot{i} (\foldParShot{i} - \intDim + 1) + \intDim (k - 1) + 1}{\intDim + 1}
        \right\rceil.
    \end{equation}
\end{lemma}

\begin{proof}
    A nontrivial solution of~\eqref{eq:intSystem} exists if less equations than unknowns are involved.
    The number of equations corresponds to the number of interpolation points and hence the condition on the existence
    of a nonzero solution reads as follows:
    \begin{align}
        \sum_{i=1}^{\shots} \lenFLRSshot{i} (\foldParShot{i} - s + 1) &< \degConstraint (\intDim + 1) - \intDim (k - 1)
        \label{eq:degConstraintIneq}
        \\
        \iff \degConstraint &\geq
        \frac{\sum_{i=1}^{\shots} \lenFLRSshot{i} (\foldParShot{i} - s + 1) + \intDim (k - 1) + 1}{\intDim + 1}.
    \end{align}
    Since $\degConstraint$ is integral, the statement follows.
\end{proof}

If the same folding parameter is used for each block, that is if there is a $h \in \NN^{\ast}$ such that
$\foldPar = \foldParShot{i}$ holds for all $1 \leq i \leq \shots$, this reduces to
\begin{equation}
    \degConstraint=\left\lceil\frac{\lenFLRS(\foldPar-\intDim+1)+\intDim(k-1)+1}{\intDim+1}\right\rceil,
\end{equation}
which coincides with~\cite[Lemma~2]{hoermann2022efficient}.
Note that this is still true for different numbers of columns $\lenFLRSshot{1}, \dots, \lenFLRSshot{\shots}$.

\begin{lemma}[Roots of Polynomial]\label{lem:decConditionFLRS}
    Define the univariate skew polynomial
    \begin{equation} \label{eq:def_univariate_skew_poly}
        \begin{aligned}
            P(x) &\defeq Q_0(x)+Q_1(x)f(x)+Q_2(x)f(x)\pe+\dots+Q_\intDim(x)f(x)\pe^{\intDim - 1} \\
            &= Q(x, f(x), f(x)\pe, \dots, f(x)\pe^{\intDim-1}) \in \SkewPolyring
        \end{aligned}
    \end{equation}
    and write $t_i \defeq \rk_q(\mat{E}^{(i)})$ for $1 \leq i \leq \shots$.
    Then there exist $\Fq$-linearly independent elements
    $\zeta_1^{(i)}, \dots, \zeta_{(\lenFLRSshot{i} - t_i)(\foldParShot{i} - \intDim + 1)}^{(i)} \in \Fqm$ for each
    $i \in \{1, \dots, \shots\}$ such that $\opev{P}{\zeta_j^{(i)}}{a_i} = 0$ for all $1 \leq i \leq \shots$ and all
    $1 \leq j \leq (\lenFLRSshot{i} - t_i) (\foldParShot{i} - \intDim + 1)$.
\end{lemma}

\begin{proof}
    Since $\rk_q(\mat{E}^{(i)}) = t_i$, there exists a nonsingular matrix $\mat{T}_i \in
    \Fq^{\lenFLRSshot{i} \times \lenFLRSshot{i}}$ such that $\mat{E}^{(i)}\mat{T}_i$ has only $t_i$ nonzero columns for
    every $i \in \{1, \ldots, \shots\}$.
    Without loss of generality assume that these columns are the last ones of $\mat{E}^{(i)}\mat{T}_i$ and define
    $\veczeta^{(i)} = \L \cdot \mat{T}_i$ with $\L \in \Fqm^{\foldParShot{i} \times \lenFLRSshot{i}}$ containing the
    code locators $1, \dots, \pe^{\lenShot{i}-1}$ (cp.~\eqref{eq:defFLRScodeblock}).
    Note that the first $\lenFLRSshot{i} - t_i$ columns of $\R^{(i)} \mat{T}_i = \C^{(i)}\mat{T}_i + \E^{(i)}\mat{T}_i$
    are noncorrupted leading to $(\lenFLRSshot{i} - t_i)(\foldParShot{i} - \intDim + 1)$ noncorrupted interpolation
    points according to~\eqref{eq:intPointDef}.
    Now, for each $1 \leq i \leq \shots$, the first entries of the
    $(\lenFLRSshot{i} - t_i)(\foldParShot{i} - \intDim + 1)$ noncorrupted interpolation points (i.e.\ the top left
    submatrix of size $(\lenFLRSshot{i} - t_i) \times (\foldParShot{i} - \intDim + 1)$ of $\zeta^{(i)}$) are by
    construction both $\Fq$-linearly independent and roots of $P(x)$.
\end{proof}

\begin{theorem}[Decoding Radius] \label{thm:decodingRadius}
    Let $Q(x,y_1,\dots,y_\intDim)$ be a nonzero solution of \autoref{prob:intProblemFLRS}.
    If the error-weight decomposition $\t = (t_1, \dots, t_{\shots})$ satisfies
    \begin{equation}\label{eq:listDecRegionFLRS}
        \sum_{i=1}^{\shots} t_i (\foldParShot{i} - \intDim + 1)
        < \frac{\intDim}{\intDim+1} \left( \sum_{i=1}^{\shots} \lenFLRSshot{i} (\foldParShot{i} - s + 1) - k + 1 \right),
    \end{equation}
    then $P \in \SkewPolyring$ defined in~\eqref{eq:def_univariate_skew_poly} is the zero polynomial, that is for all $x \in \Fqm$
    \begin{equation}\label{eq:rootFindingEquationFLRS}
        P(x)=Q_0(x)+Q_1(x)f(x)+\!\cdots\!+Q_\intDim(x)f(x)\pe^{\intDim-1}=0.
    \end{equation}
\end{theorem}

\begin{proof}
    By~\autoref{lem:decConditionFLRS}, there exist elements
    $\zeta_1^{(i)}, \dots, \zeta_{(\lenFLRSshot{i} - t_i)(\foldParShot{i} - \intDim + 1)}^{(i)}$ in $\Fqm$ that are
    $\Fq$-linearly independent for each $i \in \{1, \dots, \shots\}$ such that $\opev{P}{\zeta_j^{(i)}}{a_i} = 0$ for
    $1 \leq i \leq \shots$ and $1 \leq j \leq (\lenFLRSshot{i} - t_i)(\foldParShot{i} - \intDim + 1)$.
    By choosing
    \begin{equation} \label{eq:degConstExceedsBound}
        \degConstraint \leq \sum_{i=1}^{\shots} (\lenFLRSshot{i} - t_i)(\foldParShot{i} - \intDim + 1),
    \end{equation}
    $P(x)$ exceeds the degree bound from~\cite[Prop.~1.3.7]{caruso2019residues} which is possible only if $P(x)=0$.
    Together with inequality~\eqref{eq:degConstraintIneq}, we get
    \begin{align}
        &\phantom{\iff}
        \frac{\sum_{i=1}^{\shots} \lenFLRSshot{i} (\foldParShot{i} - s + 1) + \intDim (k - 1)}{\intDim + 1}
        < \degConstraint
        \leq \sum_{i=1}^{\shots} (\lenFLRSshot{i} - t_i)(\foldParShot{i} - \intDim + 1)
        \\
        &\iff
        \sum_{i=1}^{\shots} \lenFLRSshot{i} (\foldParShot{i} - s + 1) + \intDim (k - 1)
        < (\intDim + 1) \sum_{i=1}^{\shots} (\lenFLRSshot{i} - t_i)(\foldParShot{i} - \intDim + 1)
        \\
        &\iff
        \sum_{i=1}^{\shots} t_i (\foldParShot{i} - \intDim + 1)
        < \frac{\intDim}{\intDim+1} \left( \sum_{i=1}^{\shots} \lenFLRSshot{i} (\foldParShot{i} - s + 1) - k + 1 \right).
    \end{align}
\end{proof}

Note that the left-hand side equals the number of erroneous interpolation points.
Intuitively speaking, a rank error in a block with many rows is worse than one in a block with a small folding parameter
because it creates more corrupted interpolation points.
This is due to the fact that we can interpret a rank error in the $i$-th block as a symbol error over the extension
field $\F_{q^{\foldParShot{i}}}$ corresponding to $\foldParShot{i}$ $\Fq$-errors.

Even though~\autoref{thm:decodingRadius} describes the admissible decoding radius, the derived condition does not only
depend on the sum-rank weight $t$ of the error but also on its weight decomposition $\t = (t_1, \dots, t_{\shots})$.
If we focus on the simpler special case of using the same folding parameter $h \in \NN^{\ast}$ for all blocks,
formula~\eqref{eq:listDecRegionFLRS} simplifies to the same inequality as in~\cite[Thm.~1]{hoermann2022efficient} that
only depends on the error weight $t$.
Namely,
\begin{equation} \label{eq:decoding_region_same_h}
    t < \frac{\intDim}{\intDim+1} \left(\frac{\lenFLRS(\foldPar-\intDim+1)-k+1}{\foldPar-\intDim+1}\right).
\end{equation}
This yields the desirable property that we can characterize all decodable errors simply as the ones lying in a sum-rank
ball.

In this case, we derive the normalized decoding radius $\tau \defeq \frac{t}{\lenFLRS}$
from~\eqref{eq:decoding_region_same_h} as
\begin{equation} \label{eq:normalized_decoding_radius}
    \begin{aligned}
        \tau = \frac{t}{\lenFLRS} &< \frac{\intDim}{\intDim+1} \left( \frac{\lenFLRS (\foldPar - \intDim + 1) - k + 1}
        {\lenFLRS (\foldPar - \intDim + 1)} \right) \\
        &= \frac{\intDim}{\intDim+1} \left( 1 - \frac{\foldPar R - \frac{1}{\lenFLRS}}{\foldPar - \intDim + 1} \right)
        \xrightarrow{\lenFLRS \to \infty}
        \frac{\intDim}{\intDim+1} \left( 1 - \frac{\foldPar}{\foldPar - \intDim + 1} R \right)
    \end{aligned}
\end{equation}
where $R \defeq \frac{k}{\foldPar \lenFLRS}$ denotes the code rate.

In the more general case, we can imagine the set of decodable error patterns as a sum-rank ball with some additional
bulges.
We can derive from~\eqref{eq:normalized_decoding_radius} that all errors with sum-rank weight $t$ satisfying
\begin{equation} \label{eq:decoding_region_worst_case}
    t < \frac{\intDim}{\intDim+1}
    \frac{\sum_{i=1}^{\shots} \lenFLRSshot{i} (\foldParShot{i} - s + 1) - k + 1}{\foldParShot{max} - \intDim + 1}
\end{equation}
for $\foldParShot{max} \defeq \max_{i \in \{1, \dots, \shots\}} \foldParShot{i}$ can be decoded for sure.
This corresponds to the ball.
On the other hand, the buldges represent specific decodable error-weight decompositions having larger sum-rank weight.
However, the worst-case bound
\begin{equation} \label{eq:decoding_region_best_case}
    t < \frac{\intDim}{\intDim+1}
    \frac{\sum_{i=1}^{\shots} \lenFLRSshot{i} (\foldParShot{i} - s + 1) - k + 1}{\foldParShot{min} - \intDim + 1}
\end{equation}
with $\foldParShot{min} \defeq \min_{i \in \{1, \dots, \shots\}} \foldParShot{i}$ shows that the code can definitely
not correct error patterns of weight $t$ exceeding its right-hand side.
\autoref{tab:decodable_error_patterns} contains some examples for codes and the error patterns they can decode.

{
\rowcolors{2}{lightgray}{white}

\begin{table}
    \centering
    \setlength{\tabcolsep}{4pt}
    \caption{Decodable error-weight decompositions for codes of dimension $k = 2$ and decoder parameter $\intDim = 2$.}
    \label{tab:decodable_error_patterns}
    \begin{tabular}{|c|c||M{.9cm}|M{.9cm}|M{.9cm}|M{.9cm}|M{.9cm}||c|c|c|}
        \hline
        $\lenVec$ & $\foldParVec$ & \multicolumn{5}{c||}{number of decodable error patterns} & decoding & minimum \\
        & & $t = 1$ & $t = 2$ & $t = 3$ & $t = 4$ & $t = 5$ & radius$^{\ast}$ & distance \\
        \hline
        \hline
        $(6, 6)$ & $(3, 3)$ & 2 / 2 & 3 / 3 & none & none & none & 2.33 & 4 \\
        & $(2, 2)$ & 2 / 2 & 3 / 3 & none & none & none & 3.33 & 6 \\
        & $(3, 2)$ & 2 / 2 & 2 / 3 & 1 / 3 & none & none & 2.0~$\vert$~4.0 & 5 \\
        \hline
        \hline
        $(6, 6, 6)$ & $(3, 3, 3)$ & 3 / 3 & 6 / 6 & 7 / 7 & none & none & 3.67 & 6 \\
        & $(2, 2, 2)$ & 3 / 3 & 6 / 6 & 10 / 10 & 12 / 12 & 12 / 12 & 5.33 & 9 \\
        & $(3, 3, 2)$ & 3 / 3 & 6 / 6 & 8 / 8 & 5 / 8 & none & 3.33~$\vert$~6.67 & 7 \\
        & $(3, 2, 2)$ & 3 / 3 & 6 / 6 & 9 / 9 & 7 / 10 & 2 / 9 & 3.0~$\vert$~6.0 & 8 \\
        \hline
    \end{tabular}
    \\
    $^{\ast}$~\eqref{eq:decoding_region_same_h} for codes with the same folding parameter for all blocks,
    ~\eqref{eq:decoding_region_worst_case}~$\vert$~\eqref{eq:decoding_region_best_case} in all other cases.
\end{table}
}

\subsection{Root-Finding Step} \label{subsec:root-finding-step}

By~\autoref{thm:decodingRadius}, the message polynomial $f \in \SkewPolyring_{<k}$
satisfies~\eqref{eq:rootFindingEquationFLRS} if~\eqref{eq:listDecRegionFLRS} holds for the error-weight decomposition
$\t = (t_1, \dots, t_{\shots})$.
Therefore, we consider the following root-finding problem.
\begin{problem}[Root-Finding Problem] \label{prob:rootFinding}
    Let $Q \in \MultSkewPolyring$ be a nonzero solution of~\autoref{prob:intProblemFLRS} and let the error-weight
    decomposition $\t = (t_1, \dots, t_{\shots})$ satisfy constraint~\eqref{eq:listDecRegionFLRS}.
    Find all skew polynomials $f \in \SkewPolyring_{<k}$ for which~\eqref{eq:rootFindingEquationFLRS} applies.
\end{problem}
Condition~\eqref{eq:rootFindingEquationFLRS} is equivalent to all coefficients of the polynomial on the left-hand side
of~\eqref{eq:rootFindingEquationFLRS} being zero.
Multiple application of $\aut^{-1}$ to the equations resulting from the coefficients allows to
express~\autoref{prob:rootFinding} as an $\Fqm$-linear system of equations in the unknown
\begin{equation}
    \f \defeq (f_0, \aut^{-1}(f_1), \ldots, \aut^{-k+1}(f_{k-1}))^{\top}.
\end{equation}

As e.g.\ in~\cite{wachter2013decoding,BartzSidorenko_FoldedGabidulin2015_DCC}, we use a basis of the interpolation
problem's solution space instead of choosing only one solution $Q$ of system~\eqref{eq:intSystem}.
This improvement is justified by the following result.
\begin{lemma}[Number of Interpolation Solutions] \label{lem:intSolutionDim}
    For $d_I \defeq \dim_{q^m} (\ker(\S))$ with $\S$ defined in~\eqref{eq:intSystem}, it holds
    \begin{equation}
        d_I \geq \intDim (\degConstraint - k + 1) - \sum_{i=1}^{\shots} t_i (\foldParShot{i} - \intDim + 1).
    \end{equation}
\end{lemma}

\begin{proof}
    The first $\degConstraint$ columns of $\S$ are given as $\left(\opMoore{\degConstraint}{\p_0}{\a}\right)^{\top}$.
    Since the $\shots$ blocks of $\p_0$ consist of pairwise distinct powers of $\pe$, the elements of a single block are
    $\Fq$-linearly independent.
    Hence $\rk_{q^m}\left(\opMoore{\degConstraint}{\p_0}{\a}\right)
    = \min (\degConstraint, \sum_{i=1}^{\shots} \lenFLRSshot{i} (\foldParShot{i} - \intDim + 1)) = \degConstraint$,
    where the last equality follows from equation~\eqref{eq:degConstExceedsBound}.
    In the absence of an error, the remaining columns consist of linear combinations of the first $\degConstraint$ ones
    and do not increase the rank.
    If an error $\E$ with $\SumRankWeight(\E) = t$ is introduced, at most
    $\sum_{i=1}^{\shots} t_i (\foldParShot{i} - \intDim + 1)$ interpolation points are corrupted according
    to~\autoref{lem:decConditionFLRS}.
    As a consequence, these columns can increase the rank of $\S$ by at most
    $\sum_{i=1}^{\shots} t_i (\foldParShot{i} - \intDim + 1)$.
    Thus, $\rk_{q^m}(\S) \leq \degConstraint + \sum_{i=1}^{\shots} t_i (\foldParShot{i} - \intDim + 1)$ and the
    rank-nullity theorem directly yields
    \begin{align}
        d_I &\defeq \dim_{q^m} (\ker(\S)) = \degConstraint (\intDim + 1) - \intDim (k-1) - \rk_{q^m}(\S) \\
        &\geq \degConstraint (\intDim + 1) - \intDim (k-1) -
        \left( \degConstraint + \sum_{i=1}^{\shots} t_i (\foldParShot{i} - \intDim + 1) \right) \\
        &= \intDim (\degConstraint - k + 1) - \sum_{i=1}^{\shots} t_i (\foldParShot{i} - \intDim + 1).
    \end{align}
\end{proof}

Again we get a simpler bound when we consider the same folding parameter $h \in \NN^{\ast}$ for each block
(cp.~\cite[Lemma~4]{hoermann2022efficient}):
\begin{equation}
    d_I \geq \intDim (\degConstraint - k + 1) - t (\foldPar - \intDim + 1).
\end{equation}

Let now $Q^{(1)}, \ldots, Q^{(d_I)} \in \MultSkewPolyring$ form a basis of the solution space
of~\autoref{prob:intProblemFLRS} and denote the coefficients of $Q^{(u)}$ by $q_{i,j}^{(u)}$ for all $1 \leq u \leq d_I$
(cp.~\eqref{eq:Qcoefficients}).
In other words, we write for all $u \in \{1, \ldots, d_I\}$
\begin{equation}
    Q^{(u)}(x, y_1, \ldots, y_\shots)
    = \sum_{j=0}^{\degConstraint-1} q_{0, j}^{(u)} x^j
    + \Bigg( \sum_{j=0}^{\degConstraint-k} q_{1, j}^{(u)} x^j \Bigg)y_1 + \ldots
    + \Bigg( \sum_{j=0}^{\degConstraint-k} q_{\intDim, j}^{(u)} x^j \Bigg)y_\intDim.
\end{equation}
Define further the ordinary polynomials
\begin{equation}\label{eq:rf_polys}
    B_j^{(u)}(x) = q_{1, j}^{(u)} + q_{2, j}^{(u)} x + \cdots + q_{\intDim, j}^{(u)} x^{\intDim-1} \in \Polyring
\end{equation}
for $j \in \{0, \ldots, \degConstraint-k\}$ and $u \in \{1, \ldots, d_I\}$ as well as the additional notations
\vspace*{-5pt}
\begin{align}
    \b_{j,a} &= \left( \aut^{-a}\left(B_j^{(1)}(\aut^{a}(\pe))\right), \ldots,
    \aut^{-a}\left(B_j^{(d_I)}(\aut^{a}(\pe))\right) \right)^{\top} \\
    \text{and} \qquad
    \q_a &= \left( \aut^{-a}\left(q_{0,a}^{(1)}\right), \ldots, \aut^{-a}\left(q_{0,a}^{(d_I)}\right) \right)^{\top}
\end{align}
for $0 \leq j \leq \degConstraint-k$ and $0 \leq a \leq \degConstraint-1$.
Then the root-finding system is given as
\begin{gather}\label{eq:improvedRootFindingSystemFLRS}
    \B \cdot \f = -\q
    \\
    \text{with }
    \B \defeq
    \begin{pmatrix}
        \b_{0,0} & & & \\
        \b_{1,1} & \b_{0,1} & & \\[-3pt]
        \vdots & \b_{1,2} & \ddots & \\[-3pt]
        \b_{\degConstraint-k,\degConstraint-k} & \vdots & & \b_{0,k-1} \\
        & \b_{\degConstraint-k,\degConstraint-k+1} & & \b_{1,k} \\
        & & \ddots & \vdots \\[-3pt]
        & & & \b_{\degConstraint-k,\degConstraint-1}
    \end{pmatrix}
    \text{ and }
    \q \defeq
    \begin{pmatrix}
        \q_0 \\
        \vdots \\
        \q_{\degConstraint-1}
    \end{pmatrix}.
    \notag
\end{gather}
The root-finding system~\eqref{eq:improvedRootFindingSystemFLRS} can be solved by back substitution in at most
$\OCompl{k^2}$ operations in $\Fqm$ since we can focus on (at most) $k$ nontrivial equations from different blocks of
$d_I$ rows.
Observe that the transmitted message polynomial $f \in \SkewPolyring_{<k}$ is always a solution of~\eqref{eq:improvedRootFindingSystemFLRS}
as long as $\t = (t_1, \dots, t_{\shots})$ satisfies the decoding radius in~\eqref{eq:listDecRegionFLRS}.

\begin{example}
    Let us set up the root-finding problem for the \ac{FLRS} code considered in \autoref{ex:int_points}
    and~\autoref{ex:int_matrix}.
    We obtain
    \begin{align}
        P(x) &= \underbrace{q_{0, 0} + q_{0, 1}x + q_{0, 2}x^2}_{Q_0(x)}
        + \underbrace{(q_{1, 0} + q_{1, 1}x)(f_0 + f_1 x)}_{Q_1(x)f(x)}
        \\
        &\phantom{=} + \underbrace{(q_{2,0} + q_{2, 1}x)(f_0\pe + f_1\aut(\pe)x)}_{Q_2(x)f(x)\pe}
        \\
        \text{with} \quad Q_1(x)f(x) &= q_{1, 0}f_0 + \left( q_{1, 0}f_1 + q_{1, 1}\aut(f_0) \right)x
        + q_{1, 1}\aut(f_1) x^2
        \\
        \text{and} \quad Q_2(x)f(x)\pe &= q_{2, 0}f_0\pe
        + \left( q_{2, 0}f_1\aut(\pe) + q_{2, 1}\aut(f_0)\aut(\pe) \right)x + q_{2, 1}\aut(f_1)\aut^2(\pe) x^2.
    \end{align}
    Now we define $B_j(x) \defeq q_{1, j} + q_{2, j}x \in \Polyring$ for $j \in \{0, 1, 2\}$ and write the coefficients
    of $P(x) = p_0 + p_1 x + p_2 x^2$ as follows:
    \begin{align}
        p_0 &= q_{0, 0} + f_0(q_{1, 0} + q_{2, 0} \pe)
        = q_{0, 0} + f_0 B_0(\pe) \\
        p_1 &= q_{0, 1} + f_1(q_{1, 0}+ q_{2, 0} \aut(\pe)) + \aut(f_0)(q_{1, 1} + q_{2, 1} \aut(\pe)) \\
        &= q_{0, 1} + f_1 B_0(\aut(\pe)) + \aut(f_0) B_1(\aut(\pe)) \\
        p_2 &= q_{0, 2} + \aut(f_1)(q_{1, 1} + q_{2, 1} \aut^2(\pe))
        = q_{0, 2} + \aut(f_1) B_1(\aut^2(\pe))
    \end{align}
    Because $p_i = 0$ for all $i \in \{0, 1, 2\}$, application of $\aut^{-i}$ to the equation belonging to $p_i$ does
    not change the solution space of the above system of equations.
    Hence, we can equivalently solve the root-finding system
    \begin{equation}
        \begin{pmatrix}
            B_0(\pe) & \\
            \aut^{-1}\left( B_1(\aut(\pe)) \right) & \aut^{-1}\left( B_0(\aut(\pe)) \right) \\
            & \aut^{-2}\left( B_1(\aut^2(\pe)) \right) \\
        \end{pmatrix}
        \cdot
        \begin{pmatrix}
            f_0 \\
            \aut^{-1}(f_1)
        \end{pmatrix}
        = -
        \begin{pmatrix}
            q_{0, 0} \\
            \aut^{-1}(q_{0, 1}) \\
            \aut^{-2}(q_{0, 2})
        \end{pmatrix}.
    \end{equation}
\end{example}

\subsection{List and Probabilistic Unique Decoding} \label{subsec:list-and-probabilistic-decoding}

The interpolation-based scheme from above is summarized in~\autoref{alg:decoder} and can be used for list decoding or as
a probabilistic unique decoder.
In the first case, all solutions of~\eqref{eq:improvedRootFindingSystemFLRS} are returned as a list of candidate message
polynomials.
Note that this list contains all message polynomials corresponding to codewords having sum-rank distance less than the
decoding radius from the actually sent codeword.
However, there may also be some candidates in the list that lie outside of the sum-rank ball around the sent codeword.
In the second case, the decoder returns the unique solution of~\eqref{eq:improvedRootFindingSystemFLRS} or declares a
decoding failure if there are multiple candidates.
Let us investigate the usage of our decoding scheme as list decoder and bound its output size as follows:

\begin{algorithm}[ht!]
    \caption{\algoname{Interpolation-Based Decoding of \ac{FLRS} Codes}}
    \label{alg:decoder}
    \begin{algorithmic}[1]
        \State{Choose $s \leq \min_{i \in \{1, \dots, \shots\}} \foldParShot{i}$ and $\degConstraint$ according to~\eqref{eq:deg_constraint}}
        \State{Compute the sets $\set{P}_{1}, \dots, \set{P}_{\shots}$ of interpolation points according to~\eqref{eq:intPointDef}}
        \State{Find a basis $(Q^{(1)}, \dots, Q^{(d_I)})$ of the solution space of~\autoref{prob:intProblemFLRS}}
        \State{Find a basis $(f^{(1)}, \dots, f^{(d_{RF})})$ of the solution space of~\autoref{prob:rootFinding} with respect to all \newline skew polynomials $Q^{(1)}, \dots, Q^{(d_I)}$}
    \end{algorithmic}
\end{algorithm}

\begin{lemma}[Worst-Case List Size] \label{lem:worstCaseListSize}
    The list size is upper bounded by $q^{m(\intDim-1)}$.
\end{lemma}

\begin{proof}
    With $d_{RF} \defeq \dim_{q^m}(\ker(\B))$, the list size equals $q^{m \cdot d_{RF}}$ and $d_{RF} = k - \rk_{q^m}(\B)$
    due to the rank-nullity theorem.
    Let $\B_{\triangle}$ denote the lower triangular matrix consisting of the first $d_I k$ rows of $\B$.
    Then, $\rk_{q^m}(\B) \geq \rk_{q^m}(\B_{\triangle})$ and the latter is lower bounded by the number of nonzero vectors
    on its diagonal.
    These vectors are $\b_{0,0}, \ldots, \b_{0,k-1}$ and we focus on their first components while neglecting application
    of $\aut$.
    Observe that each of them is given as the evaluation of $B_0^{(1)}$ at another conjugate of $\pe$.
    Since $B_0^{(1)}$ can have at most $\intDim - 1$ roots, it follows that at most $\intDim - 1$ of the vectors on the
    diagonal can be zero.
    Thus, $\rk_{q^m}(\B) \geq k - \intDim + 1$ and, as a consequence, $d_{RF} \leq \intDim - 1$.
\end{proof}

Note that, despite the exponential worst-case list size, an $\Fqm$-basis of the list can be found in polynomial time.
\autoref{thm:listDecoding} summarizes the results for list decoding of \ac{FLRS} codes.

\begin{theorem}[List Decoding] \label{thm:listDecoding}
    Consider an \ac{FLRS} code $\foldedLinRS{\pe, \a, \foldParVec; \lenFLRSVec, k}$ and a codeword
    $\C$ that is transmitted over a sum-rank channel.
    Assume that the error has weight $t$ and that its weight decomposition
    $\t = (t_1, \dots, t_{\shots})$ satisfies
    \begin{equation}
       \sum_{i=1}^{\shots} t_i (\foldParShot{i} - \intDim + 1)
        < \frac{\intDim}{\intDim+1} \left( \sum_{i=1}^{\shots} \lenFLRSshot{i} (\foldParShot{i} - s + 1) - k + 1 \right)
    \end{equation}
    for an interpolation parameter $1 \leq \intDim \leq \min_{i \in \{1, \dots, \shots\}} \foldParShot{i}$.
    Then, a basis of an at most $(\intDim - 1)$-dimensional $\Fqm$-vector space that contains candidate message polynomials
    satisfying~\eqref{eq:rootFindingEquationFLRS} can be obtained in at most $\OCompl{\intDim \len^2}$ operations in $\Fqm$.
\end{theorem}

Recall that the decoding radius can be described by~\eqref{eq:decoding_region_same_h} if the same folding parameter
$\foldPar$ is used for all blocks.

A different concept is probabilistic unique decoding where the decoder either returns a unique solution or declares a
failure.
In our setting, a failure occurs exactly when the root-finding matrix $\B$ is rank-deficient.
Similar to~\cite{BartzSidorenko_FoldedGabidulin2015_DCC}, we now derive a heuristic upper bound on the probability
$\mathbb{P}\left( \rk_{q^m}(\B) < k \right)$.

\begin{lemma}[Decoding Failure Probability] \label{lem:failureProbabilityBound}
    Assume that the coefficients of the polynomials $B_0^{(u)}(x) \in \Polyring$ from~\eqref{eq:rf_polys} for
    $u \in \{1, \ldots, d_I\}$ are independent and have a uniform distribution among $\Fqm$.
    Then it holds that
    \begin{equation}\label{eq:heuristic_upper_bound}
        \mathbb{P}\left( \rk_{q^m}(\B) < k \right) \lesssim k \cdot \left( \frac{k}{q^m} \right)^{d_I},
    \end{equation}
    where $\lesssim$ indicates that the bound is a heuristic approximation.
\end{lemma}

\begin{proof}
    Let $\B_{\triangle}$ denote the matrix consisting of the first $d_I k$ rows of $\B$ as in the proof
    of~\autoref{lem:worstCaseListSize}.
    Note that
    \begin{equation}
        \mathbb{P}\left( \rk_{q^m}(\B) < k \right) \leq \mathbb{P}\left( \rk_{q^m}(\B_{\triangle}) < k \right)
    \end{equation}
    holds and we can focus on finding an upper bound for the right-hand side.
    As lower triangular matrix, $\B_{\triangle}$ has rank $k$ if and only if $\b_{0, 0}, \ldots, \b_{0, k-1}$ are
    nonzero.
    Instead of these vectors, we investigate
    $\tilde{\b}_{0,a} = \left( B_0^{(1)}(\aut^{a}(\pe)), \ldots, B_0^{(d_I)}(\aut^{a}(\pe)) \right)^{\top}$ for
    $a \in \{0, \ldots, k-1\}$ because application of $\aut$ can be neglected.
    Following ideas from~\cite[Lemma 8]{BartzSidorenko_FoldedGabidulin2015_DCC}, we can now interpret the vector
    consisting of the $u$-th entries of $\tilde{\b}_{0,0}, \ldots, \tilde{\b}_{0,k-1}$ for each $1 \leq u \leq d_I$ as
    a codeword of a Reed--Solomon code $\mycode{C}_{RS}$ of length $k$ and dimension $s$.
    These $d_I$ codewords have a uniform distribution with respect to the codebook of $\mycode{C}_{RS}$ due to our
    assumption on the distribution of the polynomial coefficients.
    Thanks to~\cite[eq. (1)]{cheung1989weightDistribution}, we can approximate the probability that a random codeword
    has full weight $k$ as
    \begin{equation}
     \frac{|\{ \c \in \mycode{C}_{RS} : \HammingWeight(\c) = k \}|}{|\mycode{C}_{RS}|} \approx
     \frac{|\{ \v \in \Fqm^k : \HammingWeight(\v) = k \}|}{|\Fqm^k|} = \left( 1 - \frac{1}{q^m} \right)^k.
    \end{equation}
    Let us fix an $a \in \{0, \ldots, k-1\}$ and consider $\tilde{\b}_{0,a}$.
    Then, any full-weight codeword induces a nonzero entry in $\tilde{\b}_{0,a}$ and conversely, the probability that
    one entry of $\tilde{\b}_{0,a}$ is zero is upper bounded by
    \begin{equation} \label{eq:failureProbProofOneEntryZero}
     1 - \left( 1 - \frac{1}{q^m} \right)^k < \frac{k}{q^m},
    \end{equation}
    where the estimation uses the binomial theorem.
    Due to our independence assumption for the coefficients of $B_0^{(u)}(x)$ for $u \in \{1, \ldots, d_I\}$, the
    entries of $\tilde{\b}_{0,a}$ are also independent and the probability that the whole vector is zero for a fixed
    $a$ is bounded by the $d_I$-th power of the right-hand side of~\eqref{eq:failureProbProofOneEntryZero}.
    Finally, the union bound deals with the probability that at least one vector is zero and yields
    \begin{equation}
     \mathbb{P}\left( \bigcup_{a=0}^{k-1} (\tilde{\b}_{0,a} = \0) \right)
     \leq \sum_{a=0}^{k-1} \mathbb{P}\left( \tilde{\b}_{0,a} = \0 \right)
     \lesssim \sum_{a=0}^{k-1} \left( \frac{k}{q^m} \right)^{d_I}
     = k \cdot \left( \frac{k}{q^m} \right)^{d_I}.
    \end{equation}
\end{proof}

\autoref{subsec:simulation_results} presents results that empirically verify the heuristic upper bound by Monte Carlo
simulations.
Moreover, further simulations show that the assumption that the coefficients of $B_0^{(u)}(x)$ for
$u \in \{1, \ldots, d_I\}$ are uniformly distributed is reasonable.

Let us now introduce a threshold parameter $\mu \in \NN^{\ast}$ and enforce $d_I \geq \mu$ by adapting the proof
of~\autoref{lem:degConstraintForExistence} which results in the new degree constraint
\begin{equation} \label{eq:new_deg_constraint}
    \degConstraint =
    \left\lceil \frac{\sum_{i=1}^{\shots} \lenFLRSshot{i} (\foldParShot{i} - \intDim + 1) + \intDim (k - 1) + \mu}
    {\intDim + 1}
    \right\rceil.
\end{equation}
We incorporate this threshold into the results we have shown so far and obtain \autoref{thm:probabilisticUniqueDecoding}
which provides a summary for probabilistic unique decoding of \ac{FLRS} codes.

\begin{theorem}[Probabilistic Unique Decoding] \label{thm:probabilisticUniqueDecoding}
    For an interpolation parameter $1 \leq \intDim \leq \min_{i \in \{1, \dots, \shots\}} \foldParShot{i}$ and a
    dimension threshold $\mu \in \NN^{\ast}$, transmit a codeword $\C$ of an \ac{FLRS} code $\foldedLinRS{\pe, \a, \foldParVec; \lenFLRSVec, k}$ over a sum-rank channel.
    Assume that the error weight $t$ has a weight decomposition $\t = (t_1, \dots, t_{\shots})$ satisfying
    \begin{equation} \label{eq:prob_unique_decoding_radius}
        \sum_{i=1}^{\shots} t_i (\foldParShot{i} - \intDim + 1)
        \leq \frac{\intDim}{\intDim+1}
        \left( \sum_{i=1}^{\shots} \lenFLRSshot{i} (\foldParShot{i} - \intDim + 1) - k + 1 \right)
        - \frac{\mu}{\intDim + 1}
    \end{equation}
    and that the coefficients of the polynomials $B_0^{(u)}(x)$ for $u \in \{1, \ldots, \mu\}$ are independent and uniformly distributed among $\Fqm$.
    Then, $\C$ can be uniquely recovered with complexity $\OCompl{\intDim \len^2}$ in $\Fqm$ and with an approximate
    probability of at least
    \begin{equation}
        1 - k \cdot \left( \frac{k}{q^m} \right)^{\mu}.
    \end{equation}
\end{theorem}

\begin{proof}
    The decoding radius follows when inequality~\eqref{eq:degConstraintIneq} is replaced by
    \begin{equation}
        \sum_{i=1}^{\shots} \lenFLRSshot{i} (\foldParShot{i} - s + 1)
        \leq \degConstraint (\intDim + 1) - \intDim (k - 1) - \mu
    \end{equation}
    in the proof of~\autoref{thm:decodingRadius}.
    Since the degree constraint~\eqref{eq:new_deg_constraint} enforces $d_I \geq \mu$, the heuristic upper bound on
    the failure probability from~\autoref{lem:failureProbabilityBound} attains the worst-case value for $d_I = \mu$.
    The success probability of decoding is hence at least $1 - k \cdot \left( \frac{k}{q^m} \right)^{\mu}$.

    The total complexity of $\OCompl{\intDim \len^2}$ follows since at most $\OCompl{\intDim \len^2}$ $\Fqm$-operations
    are needed to solve the interpolation problem and the solution of the root-finding problem can be computed in $\OCompl{k^2}$ operations.
\end{proof}

For the same folding parameter $\foldPar$ for each block, we get the simplified degree constraint
\begin{equation}
    \degConstraint =
    \left\lceil \frac{\lenFLRS(\foldPar - \intDim + 1) + \intDim (k - 1) + \mu}{\intDim + 1} \right\rceil
\end{equation}
and, as in~\cite[Thm.~3]{hoermann2022efficient}, the description of the decoding radius reads as
\begin{equation}
    t \leq \frac{\intDim}{\intDim + 1}
    \left( \frac{\lenFLRS(\foldPar - \intDim + 1) - k + 1}{\foldPar - \intDim + 1} \right)
    - \frac{\mu}{(\intDim + 1)(\foldPar - \intDim + 1)}.
\end{equation}

\subsection{Improved Decoding of High-Rate Codes} \label{subsec:improved_high_rate}

The normalized decoding radius $\tau \defeq \frac{t}{\lenFLRS}$ of the interpolation-based decoder for codes using the same folding parameter for all blocks, which is given in~\eqref{eq:normalized_decoding_radius}, is positive only for code rates $R < \frac{\foldPar - \intDim + 1}{\foldPar}$.
This is our motivation to now consider an interpolation-based decoder for \ac{FLRS} codes that allows to correct sum-rank errors beyond the unique decoding radius for any code rate $R > 0$.
The main idea behind this decoder is inspired by Justesen's decoder for \ac{FRS} codes~\cite[Sec.~III-B]{Guruswami2008Explicit},~\cite{brauchle2015} and the Justesen-like decoder for high-rate folded Ga\-bi\-du\-lin codes~\cite{bartz2015list,BartzSidorenko_FoldedGabidulin2015_DCC,bartz2017algebraic}.
Compared to the Guruswami--Rudra-like decoder from~\autoref{sec:decoding}, the proposed Justesen-like decoder uses additional interpolation points to improve the decoding performance for higher code rates.
In particular, we allow the sliding window of size $\intDim+1$ to wrap around to the neighboring symbols (except for the last symbol in each block).

As before we choose an interpolation parameter $\intDim \in \NN^{\ast}$ satisfying~\eqref{eq:intDimConstraint}.
Then for each $i = 1, \dots, \shots$ we get the index set $\set{W}_i^\text{HR}$ and the corresponding interpolation-point set $\set{P}_i^\text{HR}$ as
\begin{gather}\label{eq:intPointDefJustesen}
    \begin{aligned}
         \set{W}_i^\text{HR} \defeq &\left\{ (j-1) \foldParShot{i} + l :
         j \in \{1, \ldots, \lenFLRSshot{i}-1\}, l \in \{1, \ldots, \foldParShot{i}\} \right\}
         \\
         \cup
         &\left\{ (\lenFLRSshot{i}-1) \foldParShot{i} + l :
         l \in \{1, \ldots, \foldParShot{i} - \intDim + 1\} \right\}
         \\
        \text{and} \quad
        \set{P}_i^\text{HR} \defeq &\left\{ \left( \pe^{w-1}, r_{w}^{(i)}, r_{w+1}^{(i)}, \dots, r_{w+\intDim-1}^{(i)} \right):
        w \in \set{W}_i^\text{HR} \right\}.
    \end{aligned}
\end{gather}

\begin{remark}
    The additional interpolation points for each block $i = 1, \dots, \shots$ compared to the Guruswami--Rudra-like
    decoder can be easily deduced by the equality
    \begin{equation}
        \set{W}_i^\text{HR} = \set{W}_i \cup \left\{ (j-1) \foldParShot{i} + l :
        j \in \{1, \ldots, \lenFLRSshot{i} - 1\}, l \in \{\foldParShot{i} - \intDim + 2, \ldots, \foldParShot{i}\} \right\}.
    \end{equation}
\end{remark}

\begin{example}
    When we consider the code from~\autoref{ex:int_points}, the interpolation points for the high-rate decoder are
    \begin{equation}
        \set{P}_1^\text{HR} = \set{P}_1 \cup \Big\{ (\pe^2, r_3^{(1)}, r_4^{(1)}) \Big\}
        \qquad \text{and} \qquad
        \set{P}_2^\text{HR} = \set{P}_2 \cup \Big\{ (\pe, r_2^{(2)}, r_3^{(2)}) \Big\}.
    \end{equation}
\end{example}

\begin{problem}[High-Rate Interpolation Problem] \label{prob:intProblemFLRSjustesen}
    Solve~\autoref{prob:intProblemFLRS} with the input sets $\set{P}_1^\text{HR}, \allowbreak \dots,\allowbreak \set{P}_{\shots}^\text{HR}$,
    where the $i$-th set is associated to evaluation parameter $a_i$.
\end{problem}

Since the interpolation point set $\set{P}^\text{HR} \defeq \bigcup_{i=1}^{\shots} \set{P}_i^\text{HR}$ contains
\begin{equation*}
    \vert \set{P}^\text{HR} \vert = \sum_{i=1}^{\shots} \left(\lenFLRSshot{i}\foldParShot{i} - (\intDim - 1)\right)
    = \len - \shots(\intDim - 1)
\end{equation*}
interpolation points, we get the following condition for the existence of a nonzero solution of the high-rate
interpolation problem:

\begin{lemma}[Existence]\label{lem:degConstraintForExistenceJustesen}
    A nonzero solution to~\autoref{prob:intProblemFLRSjustesen} exists if
    \begin{equation} \label{eq:deg_constraint_existence_justesen}
        \degConstraint = \left\lceil
        \frac{\sum_{i=1}^{\shots} \lenFLRSshot{i}\foldParShot{i} - \shots(\intDim-1) + \intDim (k-1) + 1}{\intDim + 1}
        \right\rceil.
    \end{equation}
\end{lemma}

\begin{proof}
    \autoref{prob:intProblemFLRSjustesen} forms a homogeneous linear system of $\sum_{i=1}^{\shots} \left(\lenFLRSshot{i}\foldParShot{i} - (\intDim - 1)\right)$ equations in $\degConstraint(\intDim + 1) - \intDim(k - 1)$ unknowns, which has a nontrivial solution if less equations than unknowns are involved.
    This is satisfied for~\eqref{eq:deg_constraint_existence_justesen}.
\end{proof}

The new choice $\set{P}^{\text{HR}}$ of interpolation points yields at least as many uncorrupted interpolation points as $\set{P}$.
Hence, we also get at least as many $\Fq$-linearly independent zeros of the corresponding univariate polynomial $P$.

\begin{lemma}[Roots of Polynomial]\label{lem:decConditionFLRSjustesen}
    Define the univariate skew polynomial
    \begin{align}
        P(x) &\defeq Q_0(x)+Q_1(x)f(x)+Q_2(x)f(x)\pe+\dots+Q_\intDim(x)f(x)\pe^{\intDim - 1} \\
        &= Q(x, f(x), f(x)\pe, \dots, f(x)\pe^{\intDim-1}) \in \SkewPolyring
    \end{align}
    and write $t_i \defeq \rk_q(\mat{E}^{(i)})$ for $1 \leq i \leq \shots$.
    Then there exist $\Fq$-linearly independent elements
    $\zeta_1^{(i)}, \dots, \zeta_{\lenFLRSshot{i} \foldParShot{i} - (\intDim-1) - t_i(\foldParShot{i} + \intDim - 1)}^{(i)} \in \Fqm$ for each
    $i \in \{1, \dots, \shots\}$ such that $\opev{P}{\zeta_j^{(i)}}{a_i} = 0$ for all $1 \leq i \leq \shots$ and all
    $1 \leq j \leq \lenFLRSshot{i} \foldParShot{i} - (\intDim-1) - t_i(\foldParShot{i} + \intDim - 1)$.
\end{lemma}

\begin{proof}
    Since $\rk_q(\mat{E}^{(i)}) = t_i$, there exists a nonsingular matrix $\mat{T}_i \in \Fq^{\lenFLRSshot{i} \times \lenFLRSshot{i}}$ such that $\mat{E}^{(i)}\mat{T}_i$ has only $t_i$ nonzero columns for every $i \in \{1, \ldots, \shots\}$.
    Without loss of generality assume that these columns are the last ones of $\mat{E}^{(i)}\mat{T}_i$ and define $\veczeta^{(i)} = \L \cdot \mat{T}_i$ with $\L \in \Fqm^{\foldParShot{i} \times \lenFLRSshot{i}}$ containing the code locators $1, \dots, \pe^{\lenShot{i}-1}$ (cp.~\eqref{eq:defFLRScodeblock}).
    Note that the first $\lenFLRSshot{i} - t_i$ columns of $\R^{(i)} \mat{T}_i = \C^{(i)}\mat{T}_i + \E^{(i)}\mat{T}_i$ are noncorrupted leading to $\lenFLRSshot{i} \foldParShot{i} - (\intDim-1) - t_i(\foldParShot{i} + \intDim - 1)$ noncorrupted interpolation points according to~\eqref{eq:intPointDefJustesen}.
    Now, for each $1 \leq i \leq \shots$, the first entries of the $\lenFLRSshot{i} \foldParShot{i} - (\intDim-1) - t_i(\foldParShot{i} + \intDim - 1)$ noncorrupted interpolation points (i.e.\ the top left submatrix of size $(\lenFLRSshot{i} - t_i) \times (\foldParShot{i} - \intDim + 1)$ of $\zeta^{(i)}$) are by construction both $\Fq$-linearly independent and roots of $P(x)$.
\end{proof}

This results in a different decoding radius, which is shown below.

\begin{theorem}[Decoding Radius] \label{thm:decodingRadiusJustesen}
    Let $Q(x,y_1,\dots,y_\intDim)$ be a nonzero solution of~\autoref{prob:intProblemFLRSjustesen}.
    If the error-weight decomposition $\t = (t_1, \dots, t_{\shots})$ satisfies
    \begin{equation}\label{eq:listDecRegionFLRSjustesen}
        \sum_{i=1}^{\shots} t_i (\foldParShot{i} + \intDim - 1)
        < \frac{\intDim}{\intDim+1} \left( \sum_{i=1}^{\shots} \lenFLRSshot{i}\foldParShot{i} - \shots(\intDim - 1) - k + 1 \right),
    \end{equation}
    then $P \in \SkewPolyring$ is the zero polynomial, that is for all $x \in \Fqm$
    \begin{equation}\label{eq:rootFindingEquationFLRSjustesen}
        P(x)=Q_0(x)+Q_1(x)f(x)+\!\cdots\!+Q_\intDim(x)f(x)\pe^{\intDim-1}=0.
    \end{equation}
\end{theorem}

\begin{proof}
    By~\autoref{lem:decConditionFLRSjustesen}, there exist elements $\zeta_1^{(i)}, \dots, \zeta_{\lenFLRSshot{i} \foldParShot{i} - (\intDim-1) - t_i(\foldParShot{i} + \intDim - 1)}^{(i)}$ in $\Fqm$ that are $\Fq$-linearly independent for each $i \in \{1, \dots, \shots\}$ such that $\opev{P}{\zeta_j^{(i)}}{a_i} = 0$ for all $1 \leq i \leq \shots$ and $1 \leq j \leq \lenFLRSshot{i} \foldParShot{i} - (\intDim-1) - t_i(\foldParShot{i} + \intDim - 1)$.
    By choosing
    \begin{equation} \label{eq:degConstExceedsBoundJustesen}
        \degConstraint \leq \sum_{i=1}^{\shots} \left(\lenFLRSshot{i} \foldParShot{i} - (\intDim-1) - t_i(\foldParShot{i} + \intDim - 1)\right)
    \end{equation}
    $P(x)$ exceeds the degree bound from~\cite[Prop.~1.3.7]{caruso2019residues} which is possible only if $P(x)=0$.
    By combining~\eqref{eq:degConstExceedsBoundJustesen} with~\eqref{eq:deg_constraint_existence_justesen} we get
    \begin{align*}
        \sum_{i=1}^{\shots} \left(\lenFLRSshot{i}\foldParShot{i} - (\intDim - 1)\right) + \intDim (k - 1)
        &< (\intDim + 1) \left(\sum_{i=1}^{\shots} \left(\lenFLRSshot{i} \foldParShot{i} - (\intDim-1) - t_i(\foldParShot{i} + \intDim - 1)\right)\right)
        \\
        \iff \qquad
        \sum_{i=1}^{\shots} t_i (\foldParShot{i} + \intDim - 1)
        &< \frac{\intDim}{\intDim+1} \left( \sum_{i=1}^{\shots} \lenFLRSshot{i}\foldParShot{i} - \shots(\intDim - 1) - k + 1 \right).
    \end{align*}
\end{proof}

For the same folding parameter $h \in \NN^{\ast}$ for all blocks the decoding radius in~\eqref{eq:listDecRegionFLRSjustesen} simplifies to
\begin{equation} \label{eq:decoding_region_justesen_same_h}
    t < \frac{\intDim}{\intDim+1} \left(\frac{\lenFLRS \foldPar -\shots(\intDim - 1)-k+1}{\foldPar+\intDim-1}\right)
\end{equation}
which for $\shots=1$ coincides with the result for high-rate folded Gabidulin codes from~\cite{bartz2015list,BartzSidorenko_FoldedGabidulin2015_DCC,bartz2017algebraic}.

Similar to~\eqref{eq:normalized_decoding_radius}, we also derive the normalized decoding radius $\tau \defeq
\frac{t}{\lenFLRS}$ for codes with the same folding parameter $\foldPar$ for each block
from~\eqref{eq:decoding_region_justesen_same_h} and obtain
\begin{align} \label{eq:normalized_decoding_radius_Justesen}
    \tau = \frac{t}{\lenFLRS} &< \frac{\intDim}{\intDim+1} \left( \frac{\lenFLRS \foldPar - \shots (\intDim - 1) - k + 1}
    {\lenFLRS (\foldPar + \intDim - 1)} \right)
    \\
    &\xrightarrow{\lenFLRS \to \infty}
    \frac{\intDim}{\intDim+1} \frac{\foldPar}{\foldPar + \intDim - 1} \left( 1 - R \right)
\end{align}
for the code rate $R \defeq \frac{k}{\foldPar \lenFLRS}$.

\autoref{thm:decodingRadiusJustesen} shows that if the weight decomposition $\t$ of the error satisfies~\eqref{eq:listDecRegionFLRSjustesen}, a list containing the message polynomial $f \in \SkewPolyring_{< k}$ can be obtained by finding all solutions of~\eqref{eq:rootFindingEquationFLRSjustesen}.
This coincides with the root-finding problem from~\autoref{subsec:root-finding-step} and we can hence summarize the list decoder for high-rate \ac{FLRS} codes as follows:

\begin{theorem}[List Decoding] \label{thm:listDecodingJustesen}
    Consider an \ac{FLRS} code $\foldedLinRS{\pe, \a, \foldParVec; \lenFLRSVec, k}$ and a codeword
    $\C$ that is transmitted over a sum-rank channel such that the error has weight $t$ and its weight decomposition
    $\t = (t_1, \dots, t_{\shots})$ satisfies
    \begin{equation}
       \sum_{i=1}^{\shots} t_i (\foldParShot{i} + \intDim - 1)
        < \frac{\intDim}{\intDim+1} \left( \sum_{i=1}^{\shots} \lenFLRSshot{i}\foldParShot{i} - \shots(\intDim - 1) - k + 1 \right)
    \end{equation}
    for an interpolation parameter $1 \leq \intDim \leq \min_{i \in \{1, \dots, \shots\}} \foldParShot{i}$.
    Then, a basis of an at most $(\intDim - 1)$-dimensional $\Fqm$-vector space that contains candidate message polynomials
    satisfying~\eqref{eq:rootFindingEquationFLRSjustesen} can be obtained in at most $\OCompl{\intDim \len^2}$ operations in $\Fqm$.
\end{theorem}

By following the ideas of~\autoref{lem:intSolutionDim} we observe that the dimension of the $\Fqm$-linear solution space of the interpolation system for the Justesen-like decoder satisfies
\begin{equation*}
    d_I \geq \intDim(\degConstraint - k + 1) - \sum_{i=1}^{\shots} t_i(\foldParShot{i} - \intDim + 1).
\end{equation*}
Imposing the threshold $d_I \geq \mu$ yields to the degree constraint
\begin{equation}
    \degConstraint = \left\lceil \frac{\sum_{i=1}^{\shots}\lenFLRSshot{i} \foldParShot{i} - \shots(\intDim -1) - \intDim(k - 1) + \mu}{\intDim + 1} \right \rceil
\end{equation}
which lets us provide a summary for probabilistic unique decoding of \ac{FLRS} codes in~\autoref{thm:probabilisticUniqueDecodingJustesen}.

\begin{theorem}[Probabilistic Unique Decoding] \label{thm:probabilisticUniqueDecodingJustesen}
    For an interpolation parameter $1 \leq \intDim \leq \min_{i \in \{1, \dots, \shots\}} \foldParShot{i}$ and a dimension threshold $\mu \in \NN^{\ast}$, transmit a codeword $\C$ of an \ac{FLRS} code $\foldedLinRS{\pe, \a, \foldParVec; \lenFLRSVec, k}$ over a sum-rank channel.
    If the coefficients of the polynomials $B_0^{(u)}(x)$ for $u \in \{1, \ldots, \mu\}$ are independent and uniformly distributed among $\Fqm$ and the error-weight decomposition
    $\t = (t_1, \dots, t_{\shots})$ satisfies
    \begin{equation}
       \sum_{i=1}^{\shots} t_i (\foldParShot{i} + \intDim - 1)
        \leq \frac{\intDim}{\intDim+1} \left( \sum_{i=1}^{\shots} \lenFLRSshot{i}\foldParShot{i} - \shots(\intDim - 1) - k + 1 \right)
        - \frac{\mu}{\intDim + 1},
    \end{equation}
    $\C$ can be uniquely recovered with complexity $\OCompl{\intDim \len^2}$ in $\Fqm$ with an approximate probability of at least
    \begin{equation}
        1 - k \cdot \left( \frac{k}{q^m} \right)^{\mu}.
    \end{equation}
\end{theorem}

For the same folding parameter $\foldPar$ for each block, we get the decoding radius
\begin{equation}
    t \leq \frac{\intDim}{\intDim + 1}
    \left( \frac{\lenFLRS\foldPar - k + 1 - \shots(\intDim - 1)}{\foldPar + \intDim - 1} \right)
    - \frac{\mu}{(\intDim + 1)(\foldPar + \intDim - 1)}
\end{equation}
which for $\shots = 1$ coincides with the probabilistic unique decoding radius for folded Gabidulin codes (cf.~\cite[Thm.~3]{BartzSidorenko_FoldedGabidulin2015_DCC}).

\autoref{fig:radiusFLRSopt} illustrates the normalized decoding radii of the presented Gu\-ru\-swa\-mi--Ru\-dra- and Justesen-like decoders for \ac{FLRS} codes.
In particular, the significant improvement upon unique decoding is shown.

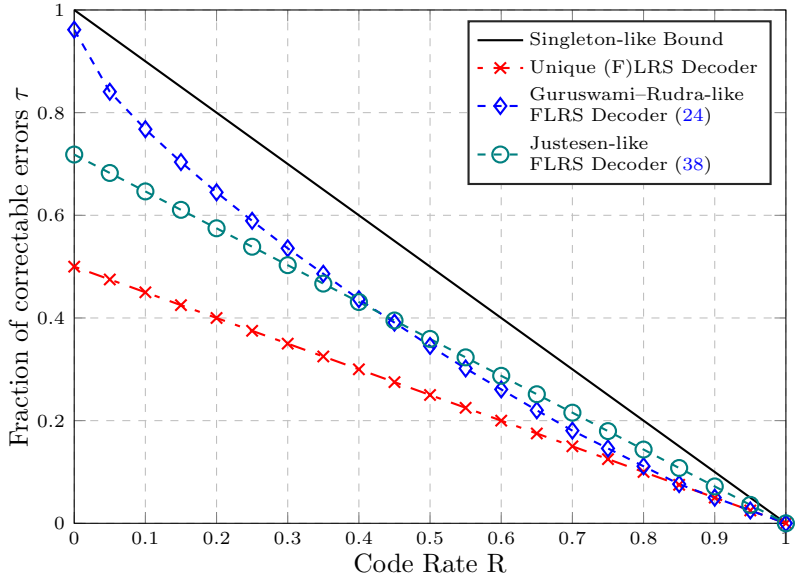
\begin{figure}[ht!]
    \centering
\begin{tikzpicture}

\begin{axis}[%
scale=1.27,
xmin=0,
xmax=1,
xlabel={Code Rate R},
xmajorgrids,
ymin=0,
ymax=1,
ylabel={Fraction of correctable errors $\tau$},
ymajorgrids,
legend style={legend cell align=left,align=left,draw=white!15!black}
]
\addplot [style=SB]
  table[row sep=crcr]{%
0	1\\
0.05	0.95\\
0.1	0.9\\
0.15	0.85\\
0.2	0.8\\
0.25	0.75\\
0.3	0.7\\
0.35	0.65\\
0.4	0.6\\
0.45	0.55\\
0.5	0.5\\
0.55	0.45\\
0.6	0.4\\
0.65	0.35\\
0.7	0.3\\
0.75	0.25\\
0.8	0.2\\
0.85	0.15\\
0.9	0.1\\
0.95	0.05\\
1	0\\
};
\addlegendentry{Singleton-like Bound}

\addplot [style=LRS]
  table[row sep=crcr]{%
0 0.5\\
0.05  0.475\\
0.1 0.45\\
0.15  0.425\\
0.2 0.4\\
0.25  0.375\\
0.3 0.35\\
0.35  0.325\\
0.4 0.3\\
0.45  0.275\\
0.5 0.25\\
0.55  0.225\\
0.6 0.2\\
0.65  0.175\\
0.7 0.15\\
0.75  0.125\\
0.8 0.1\\
0.85  0.075\\
0.9 0.05\\
0.95  0.025\\
1 0\\
};
\addlegendentry{Unique (F)LRS Decoder}

\addplot [style=FLRS]
  table[row sep=crcr]{%
0.0 0.9615384615384616 \\
0.05 0.8406593406593407 \\
0.1 0.7676470588235295 \\
0.15 0.7037037037037037 \\
0.2 0.6447368421052632 \\
0.25 0.5892857142857143 \\
0.3 0.5357142857142857 \\
0.35 0.4861111111111111 \\
0.4 0.4365079365079365 \\
0.45 0.39090909090909093 \\
0.5 0.34545454545454546 \\
0.55 0.3016304347826087 \\
0.6 0.2608695652173913 \\
0.65 0.22010869565217392 \\
0.7 0.18055555555555555 \\
0.75 0.14583333333333334 \\
0.8 0.1111111111111111 \\
0.85 0.0763888888888889 \\
0.9 0.05 \\
0.95 0.025 \\
1.0 0.0 \\
};
\addlegendentry{Guruswami--Rudra-like\\FLRS Decoder~\eqref{eq:normalized_decoding_radius}}

\addplot [style=HRFLRS]
  table[row sep=crcr]{%
0.0 0.7183908045977011 \\
0.05 0.682471264367816 \\
0.1 0.646551724137931 \\
0.15 0.610632183908046 \\
0.2 0.5747126436781609 \\
0.25 0.5387931034482758 \\
0.3 0.5028735632183907 \\
0.35 0.46695402298850575 \\
0.4 0.43103448275862066 \\
0.45 0.39511494252873564 \\
0.5 0.35919540229885055 \\
0.55 0.32327586206896547 \\
0.6 0.28735632183908044 \\
0.65 0.25143678160919536 \\
0.7 0.21551724137931036 \\
0.75 0.17959770114942528 \\
0.8 0.1436781609195402 \\
0.85 0.10775862068965518 \\
0.9 0.0718390804597701 \\
0.95 0.03591954022988509 \\
1.0 0.0 \\
};
\addlegendentry{Justesen-like\\FLRS Decoder~\eqref{eq:normalized_decoding_radius_Justesen}}

\end{axis}
\end{tikzpicture}%
    \caption{
        Normalized decoding radius $\tau \defeq \frac{t}{\lenFLRS}$ vs. code rate $R \defeq \frac{k}{\lenFLRS \foldPar}$
        for an \ac{FLRS} code with the same folding parameter $\foldPar = 25$ for each block and optimal decoding
        parameter $\intDim \leq \foldPar$ for each code rate.
    }
    \label{fig:radiusFLRSopt}
\end{figure}

\subsection{Simulation Results} \label{subsec:simulation_results}

We ran simulations\footnote{The underlying data can be shared upon reasonable request.} in SageMath~\cite{sage} to empirically verify the heuristic upper bound for probabilistic unique
decoding that we derived in~\autoref{thm:probabilisticUniqueDecoding}.
We designed the parameter sets to obtain experimentally observable failure probabilities.
Therefore, we considered codes with parameters
\begin{equation}
    q = 3, \qquad m = 6, \qquad k = 2, \qquad \lenVec = (6, 6),
\end{equation}
and with two different vectors $\foldParVec \in \{ (3, 3), (3, 2) \}$ of folding parameters.
The code using $\foldParVec = (3, 3)$ has minimum distance $4$ which implies a unique-decoding
radius of $1.5$.
In contrast, the proposed probabilistic unique decoder with $\intDim = 2$ allows to correct errors of weight $t = 2$ for $\mu
\in \{1, 2\}$.
Namely, the bound~\eqref{eq:prob_unique_decoding_radius} yields $t \leq 2.17$ for $\mu = 1$ and $t \leq 2$ for $\mu = 2$.
We investigated the case $\mu = 1$ by means of a Monte Carlo simulation and collected $100$ decoding failures within
about $4.23 \cdot 10^7$ transmissions with randomly chosen error patterns of fixed sum-rank weight $t = 2$.
This gives an observed failure probability of about $2.36 \cdot 10^{-6}$, while the heuristic yields an upper bound of
$5.49 \cdot 10^{-3}$.

For $\foldParVec = (3, 2)$, the code has the higher minimum distance $5$
and a unique-decoding radius of $2$.
Its decodable error-weight decompositions with respect to the probabilistic unique decoder with $\intDim = 2$ and
$\mu \in \{1, 2, 3\}$ are
\begin{itemize}
    \item $(0, 1)$ and $(1, 0)$, i.e.\ all possible patterns for weight $t = 1$,
    \item $(0, 2)$ and $(1, 1)$, i.e.\ two out of three possible patterns for weight $t = 2$,
    \item and $(0, 3)$, i.e.\ one out of three possible patterns for weight $t = 3$.
\end{itemize}
Note that the error patterns are not equally likely.
For example, being able to correct one out of two error-weight decompositions for a given weight does not necessarily
mean that half of all errors of the given sum-rank weight can be corrected.
We ran two Monte Carlo simulations for $t = 2$ and $t = 3$ and collected in both cases $100$ failures for $\mu = 1$.
The errors were chosen uniformly at random from the set of all vectors having the prescribed sum-rank weight as well as
a decodable weight decomposition.
The observed failure probability was $1.11 \cdot 10^{-3}$ for $t = 2$ ($100$ failures in about $9.03 \cdot 10^{4}$ runs)
and $2.11 \cdot 10^{-5}$ for $t = 3$ ($100$ failures in $4.73 \cdot 10^{6}$ runs).
In both scenarios, the heuristic upper bound is $5.49 \cdot 10^{-3}$ as for the first code.

Similar to results in~\cite{wachter2013decoding,BartzSidorenko_FoldedGabidulin2015_DCC}, our heuristic upper bound is based on the assumption that the coefficients of the polynomials $B_0^{(u)}(x)
\in \mathbb{F}_{729}[x]$ with $1 \leq u \leq \mu$ defined in~\eqref{eq:rf_polys} are uniformly distributed among
$\F_{729}$.
Unfortunately, this assumption was not backed by evidence in former work.
We thus decided to investigate experimentally observed distributions and compare them with the uniform distribution
by means of the \emph{\ac{KL} divergence}.
The \ac{KL} divergence (or \emph{relative entropy},
see~\cite[Sec.~2.3]{cover2006elementsOfInformationTheory}) is a tool to measure the distance between two probability
distributions, that is often used in coding and information theory.
Note that it is not a metric in the mathematical sense but provides sufficient insights for our purpose.
In particular, it is an upper bound for other widely used statistical distance measures as e.g.~the total variation distance~\cite{Gibbs2002probabilityMetrics}.

The \ac{KL} divergence of two probability mass functions $u(x)$ and $v(x)$, that are defined over a finite
alphabet $\set{A}$, is defined as
\begin{equation}
    D_{KL}(u \,||\, v) \defeq \sum_{x \in \set{A}} u(x) \log \left( \frac{u(x)}{v(x)} \right).
\end{equation}
We understand $0 \cdot \log( \frac{0}{q} ) \defeq 0$ for any $q$ and $p \cdot \log( \frac{p}{0} ) \defeq \infty$ for any
nonzero $p$ by convention.
We follow the common approach and consider the logarithm with base 2 and thus measure the \acl{KL} divergence in bits.
Note that the divergence is always nonnegative and it equals zero if and only if the two considered probability mass
functions are equal (see e.g.~\cite[Thm.~2.6.3]{cover2006elementsOfInformationTheory})).

Denote the observed probability mass function of the coefficients of the polynomials $B_0^{(1)}(x) \in \mathbb{F}_{729}[x]$
from~\eqref{eq:rf_polys} after $10^{6}$ transmissions by $\chi$ and let $\unif_{\F_{729}}$ be the probability mass
function of the uniform distribution among $\F_{729}$.
We obtained the \ac{KL} divergence values
\begin{itemize}
    \item $D_{KL}(\chi \,||\, \unif_{\F_{729}}) \approx 3.32 \cdot 10^{-4}$ bits for $\foldParVec = (3, 3)$ and $t = 2$,
    \item $D_{KL}(\chi \,||\, \unif_{\F_{729}}) \approx 2.30 \cdot 10^{-4}$ bits for $\foldParVec = (3, 2)$ and $t = 2$, and
    \item $D_{KL}(\chi \,||\, \unif_{\F_{729}}) \approx 3.42 \cdot 10^{-4}$ bits for $\foldParVec = (3, 2)$ and $t = 3$.
\end{itemize}
This shows that the measured distribution is in all cases remarkably close to the uniform distribution, which justifies
the assumption in~\autoref{thm:probabilisticUniqueDecoding}.
The results are illustrated in more detail in~\autoref{fig:distributions_and_divergence}, where the subfigures~(a)--(c)
show the probability mass functions $\chi$ of the coefficients that were observed within $10^{6}$ transmissions for
$\foldParVec = (3, 3)$ with error weight $t = 2$ and for $\foldParVec = (3, 2)$ with error weight $t = 2$ and $t = 3$, respectively.
The red line marks the (in fact discrete) probability mass function $\unif_{\F_{729}}$ of the uniform distribution for
reference.
Subfigure (d) shows the evolution of the \acl{KL} divergence $D_{KL}(\chi \,||\, \unif_{\F_{729}})$ over the
$10^{6}$ runs for all investigated scenarios.

\begin{figure}[ht]
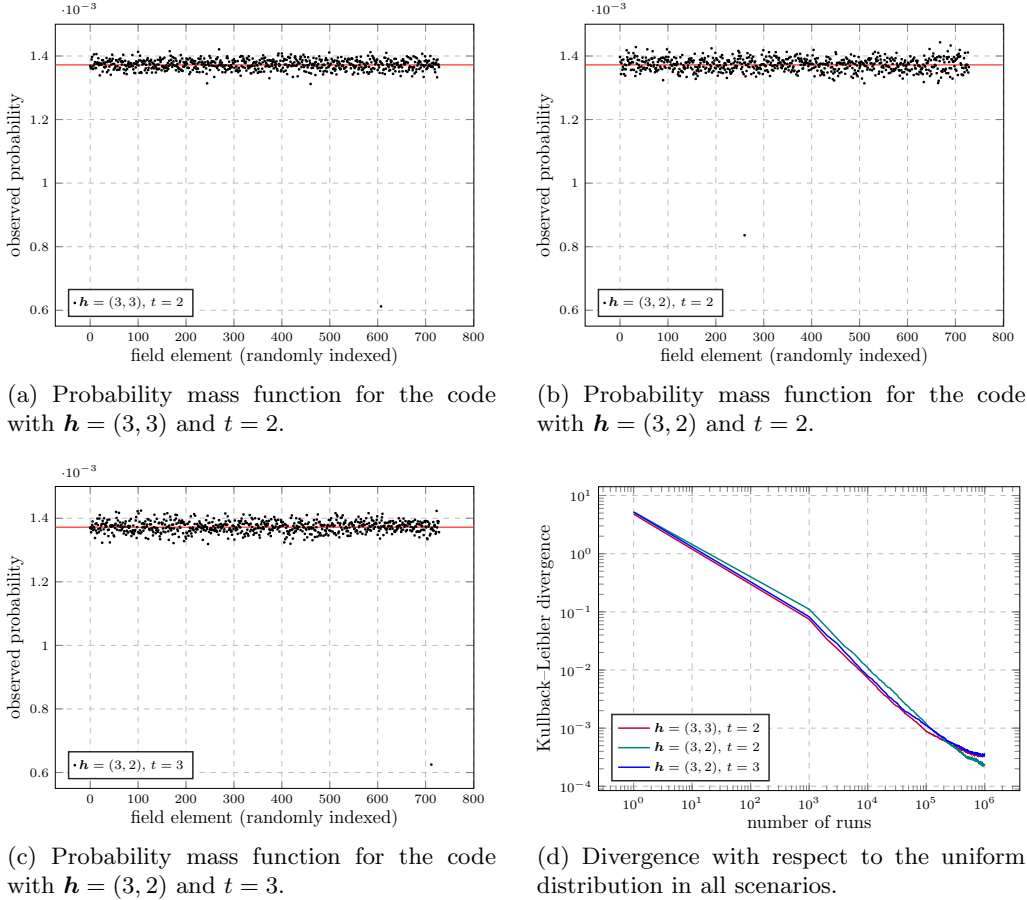

    \centering
    \begin{subfigure}[b]{.48\textwidth}
        \centering
        \resizebox{\textwidth}{!}{%

        }
        \caption{Probability mass function for the code with $\foldParVec = (3, 3)$ and $t = 2$.}
        \label{fig:distribution-1}
    \end{subfigure}
    \hfill
    \begin{subfigure}[b]{.48\textwidth}
        \centering
        \resizebox{\textwidth}{!}{%
%
        }
        \caption{Probability mass function for the code with $\foldParVec = (3, 2)$ and $t = 2$.}
        \label{fig:distribution-2}
    \end{subfigure}
    \\
    \vspace*{.3cm}
    \begin{subfigure}[t]{.48\textwidth}
        \centering
        \resizebox{\textwidth}{!}{%
%
        }
        \caption{Probability mass function for the code with $\foldParVec = (3, 2)$ and $t = 3$.}
        \label{fig:distribution-3}
    \end{subfigure}
    \hfill
    \begin{subfigure}[t]{.48\textwidth}
        \centering
        \resizebox{\textwidth}{!}{%
%
        }
        \caption{Divergence with respect to the uniform distribution in all scenarios.}
        \label{fig:divergence}
    \end{subfigure}
    \caption{Observed probability mass function of the coefficients of $B_0^{(1)}(x)$ after $10^{6}$ probabilistic
    unique decodings with $\intDim = 2$ and $\mu = 1$ for codes with $q = 3$, $m = 6$, $k = 2$, $\lenVec = (6, 6)$ and
    either $\foldParVec = (3, 3)$ and $t = 2$ or $\foldParVec = (3, 2)$ and $t \in \{2, 3\}$ and evolution of its
    divergence with respect to the uniform distribution.}
    \label{fig:distributions_and_divergence}
\end{figure}

\section{Implications for Folded Skew Reed--Solomon Codes} \label{sec:implications-fsrs}

Motivated by the relation between \ac{LRS} codes and \ac{SRS} codes from~\cite{martinez2018skew}, we now derive \ac{FSRS} codes for the skew metric from \ac{FLRS} codes.
The skew metric is related to skew evaluation codes and was introduced in~\cite{martinez2018skew}.
Decoding schemes for \ac{SRS} codes that allow for correcting errors of skew weight up to the unique-decoding radius $\lfloor\frac{n-k}{2}\rfloor$ were presented in~\cite{martinez2018skew,boucher2020algorithm,liu2015construction,bartz2020fast}.

In this section we consider decoding of \ac{FSRS} codes with respect to the \emph{(burst) skew metric}, which was introduced for \ac{ISRS} codes in~\cite{bartz2022fast}.
In the following, we restrict ourselves to evaluation codes constructed by $\SkewPolyringZeroDer$, i.e. to the zero-derivation case.

\subsection{Preliminaries on Remainder Evaluation} \label{subsec:preliminaries-on-remainder-evaluation}

Apart from the (generalized) operator evaluation (see~\autoref{sec:preliminaries}) there exists the so-called \emph{remainder evaluation} for skew polynomials, which can be seen as an analog of the classical polynomial evaluation via polynomial division.

For a skew polynomial $f\in\SkewPolyringZeroDer$ the remainder evaluation $\remev{f}{b}$ of $f$ at an element $b\in\Fqm$ is defined as the unique remainder of the right division of $f(x)$ by $(x-b)$ such that (see~\cite{lam1985general, lam1988vandermonde})
\begin{equation}\label{eq:remEval}
f(x)=g(x)(x-b)+\remev{f}{b} \quad \Longleftrightarrow\quad \remev{f}{b} = f(x) \modr (x-b).
\end{equation}
We denote the evaluation of $f$ at all entries of a vector $\vec{b} = (b_1, \dots, b_n)\in\Fqm^n$ by $\remev{f}{\vec{b}}=\left(\remev{f}{b_1},\remev{f}{b_2},\dots,\remev{f}{b_n}\right)$.

For any $a\in\Fqm$, $b\in\Fqm^*$ and $f\in\SkewPolyringZeroDer$ the generalized operator evaluation of $f$ at $b$ with respect to $a$ is related to the remainder evaluation by (see~\cite{martinez2018skew,leroy1995pseudolinear})
\begin{equation}\label{eq:rel_remev_opev}
    \remev{f}{\op{a}{b}b^{-1}}=\opev{f}{b}{a}b^{-1}.
\end{equation}

The following notions were introduced in~\cite{lam1985general,lam1988vandermonde,lam1988algebraic}, and we use the notation of~\cite{martinez2019reliable}.
Let $\set{A} \subseteq \SkewPolyringZeroDer$, $\Omega \subseteq \Fqm$, and $a \in \Fqm$.
The \emph{zero set} of $\set{A}$ is defined as
\begin{equation*}
  \set{Z}(\set{A}) := \left\{ \alpha \in \Fqm \, : \, \remev{f}{\alpha} = 0 \, \forall \, f \in \set{A} \right\}
\end{equation*}
and
\begin{equation*}
  I(\Omega) := \left\{ f \in \SkewPolyringZeroDer \, : \, \remev{f}{\alpha} = 0 \, \forall \, \alpha \in \Omega \right\}
\end{equation*}
denotes the \emph{associated ideal} of $\Omega$.
The \emph{P-closure} (or \emph{polynomial closure}) of $\Omega$ is defined by $\bar{\Omega} := \set{Z}(I(\Omega))$, and $\Omega$ is called \emph{P-closed} if $\bar{\Omega}=\Omega$.
Note that a P-closure is always P-closed.
All elements of $\Fqm \setminus \bar{\Omega}$ are said to be \emph{P-independent from $\Omega$}.
A set $\set{B} \subseteq \Fqm$ is said to be \emph{P-independent} if any $b \in \set{B}$ is P-independent from $\set{B} \setminus \{b\}$.
If $\set{B}$ is P-independent and $\Omega := \bar{\set{B}} \subseteq \Fqm$, we say that $\set{B}$ is a \emph{P-basis of $\Omega$}.

For any $\x = (x_1,\dots,x_n) \in \Fqm^n$ and any P-basis $\set{B}=\{b_1,b_2,\dots,b_n\}$, there is a unique skew polynomial $\IPrem{\set{B}, \x} \in \SkewPolyringZeroDer$ of degree less than $n$ such that
\begin{align*}
\remev{\IPrem{\set{B}, \x}}{b_j} = x_j \quad \forall \, j=1,\dots,n.
\end{align*}
We call this the \emph{remainder interpolation polynomial} of $\x$ on $\set{B}$.
The \emph{skew weight} of a vector $\vec{x}\in\Fqm^n$ with respect to a P-closed set $\Omega=\bar{\set{B}}$ with P-basis $\set{B}=\{b_1,b_2,\dots,b_n\}$\footnote{%
  Here and in the sequel, we slightly abuse notation and take this to mean $\set{B}$ is an ordered set and that the $b_i$ are distinct.%
} is defined as (see~\cite{boucher2020algorithm})
\begin{equation}\label{eq:def_skew_weight}
 \SkewWeight^\set{B}(\vec{x})\defeq\SkewWeight^\set{B}(\IPrem{\set{B}, \x})\defeq\deg\left(\lclm\left(x-\conj{b_i}{x_i}\right)_{\mystack{1\leq i\leq n}{x_i\neq0}}\right).
\end{equation}
The skew weight of a vector depends on $\Omega$ but is independent from the particular P-basis $\set{B}$ (see~\cite[Prop.~13]{martinez2018skew}) of $\Omega$.
In order to simplify the notation for the skew weights of vectors, we indicate the dependence on $\Omega$ by a particular P-basis for $\Omega$ and use the notation $\SkewWeight(\cdot)$ whenever $\set{B}$ is clear from the context.
Similar to the rank and the sum-rank weight we have that $\SkewWeight(\vec{x})\leq\HammingWeight(\vec{x})$ for all $\vec{x}\in\Fqm^n$ (see~\cite{martinez2018skew}).
The \emph{skew distance} between two vectors $\vec{x},\vec{y}\in\Fqm^n$ is defined as
\begin{equation*}
 \SkewDist(\vec{x},\vec{y})\defeq\SkewWeight(\vec{x}-\vec{y}).
\end{equation*}

\subsection{Skew Metric for Folded Matrices} \label{subsec:skew-metric-folded}

By fixing a basis of $\Fqmh$ over $\Fqm$ we can consider a matrix $\X\in\Fqm^{\foldPar\times N}$ as a vector $\vec{x}=(x_1,x_2,\dots,x_N)\in\Fqmh^N$.
Similarly, we consider a tuple $\X\in\Fqm^{\h \times \N}$ as a matrix in $\Fqm^{\foldPar \times N}$ whenever the folding parameter $\foldPar$ is the same for each block, i.e. if $\h=(\foldPar,\dots,\foldPar)$.
Similar as for \ac{ISRS} codes we define the skew weight of a matrix $\X\in\Fqm^{\foldPar\times N}$ (or a tuple $\X\in\Fqm^{\h \times \N}$) with respect to $\set{B}$ as the skew weight of the vector $\x=\extInv_{\vecgamma}(\X)\in\Fqmh^N$, i.e. as (see~\cite{bartz2022fast})
\begin{equation}\label{eq:def_skew_weight_mat}
  \SkewWeight^\set{B}(\mat{X})\defeq \SkewWeight^\set{B}(\extInv_{\vecgamma}(\mat{X})) = \deg\left(\lclm\left(x-\conj{b_i}{x_i}\right)_{\mystack{1\leq i\leq N}{x_i\neq0}}\right)
\end{equation}
where the polynomial
\begin{equation}
  \lclm\left(x-\conj{b_i}{x_i}\right)_{\mystack{1\leq i\leq N}{x_i\neq0}}
\end{equation}
is now from $\Fqmh[\aut;x]$ since we have that $x_i\in\Fqmh$ for all $i=1,\dots,N$.

\begin{lemma}\label{lem:alpha_power_remev}
  Define $\genNorm{i}{a} \defeq \prod_{k=0}^{i-1} \aut^{k}(a)$ for any $i \in \NN$ and $a \in \Fqm$, let $c \in \Fqm$ and consider a skew polynomial $f \in \SkewPolyringZeroDer$.
  Then for any $b \in \Fqm$ and $j \in \NN$ we have that
  \begin{equation*}
    \remev{f}{c^j b} = \remev{\tilde{f}}{b}
  \end{equation*}
  where $\tilde{f}=\sum_{i=0}^{\deg(f)} f_i \genNorm{i}{c^j} x^i$.
\end{lemma}

\begin{proof}
  By~\cite[Lemma~2.4]{lam1988vandermonde} we have that
  \begin{equation*}
    \remev{f}{c^j b}
    = \sum_{i=0}^{\deg(f)} f_i \genNorm{i}{c^j b}
    \overset{(*)}{=} \sum_{i=0}^{\deg(f)} f_i \genNorm{i}{c^j} \genNorm{i}{b}
    = \remev{\tilde{f}}{b}
  \end{equation*}
  where $\tilde{f}=\sum_{i=0}^{\deg(f)} f_i \genNorm{i}{c^j} x^i$ and $(*)$ follows since $f \in \SkewPolyringZeroDer$.
\end{proof}

The following result shows that each matrix can be represented as the remainder evaluation of a single skew polynomial over the large field $\Fqmh$ at evaluation points from the small field $\Fqm$.

\begin{lemma}\label{lem:folded_mat_to_interleaved_mat_fsrs}
 Let $\pe$ be a primitive element of $\Fqm$, define $\FSRSoffset \defeq \aut(\pe)/\pe$, let the folding parameter $\foldPar$ divide $\len \leq m$ and define $\lenFLRS \defeq\frac{\len}{\foldPar}$.
 Consider an evaluation parameter $a\in\Fqm^*$ and define the vector
 \begin{equation*}
    \b \defeq (\op{a}{1}/1, \op{a}{\pe^\foldPar}/\pe^\foldPar, \dots, \op{a}{\pe^{(\lenFLRS-1)\foldPar}}/\pe^{(\lenFLRS-1)\foldPar}) \in \Fqm^{\lenFLRS}.
 \end{equation*}
 Then any matrix $\X\in\Fqm^{\foldPar \times \lenFLRS}$ can be represented as
 \begin{equation*}
   \X =
   \begin{pmatrix}
    \remev{f}{\b}
    \\
    \remev{f}{\FSRSoffset\b}
    \\
    \vdots
    \\
    \remev{f}{\FSRSoffset^{\foldPar-1} \b}
   \end{pmatrix}
   =
   \begin{pmatrix}
    \remev{f^{(1)}}{\b}
    \\
    \remev{f^{(2)}}{\b}
    \\
    \vdots
    \\
    \remev{f^{(\foldPar)}}{\b}
   \end{pmatrix}
 \end{equation*}
 for some $f, f^{(1)}, \dots, f^{(\foldPar)} \in \SkewPolyringZeroDer_{<\len}$.
 Further, we have that
 \begin{equation*}
   \extInv_{\vecgamma}(\X) = \remev{F}{\b}
 \end{equation*}
 for some $F \in \Fqmh[x; \aut]_{<\len}$.
\end{lemma}

\begin{proof}
    Let $\x = \foldOpInv{\foldPar}(\X)$ be the vector obtained by unfolding $\X$ and define
    \begin{equation*}
        \widetilde{\b} \defeq \left(\op{a}{1}/1, \op{a}{\pe}/\pe, \dots, \op{a}{\pe^{(\lenFLRS-1)\foldPar}}/\pe^{(\lenFLRS-1)\foldPar}\right) \in \Fqm^{\lenFLRS \foldPar}.
    \end{equation*}
    Since $\pe$ is a primitive element of $\Fqm$ we have that the entries in $\widetilde{\b}$ are P-independent.
    Let $f \defeq \IPrem{\widetilde{\b}, \x} \in \SkewPolyringZeroDer_{< n}$ be the unique interpolation polynomial satisfying
    \begin{equation*}
     \remev{\IPrem{\widetilde{\b}, \x}}{\widetilde{b}_j} = x_j \quad \forall j=1,\dots, \foldPar \lenFLRS.
    \end{equation*}
    Due to the structure of the evaluation points in $\widetilde{\b}$ we can write $\X$ as
    \begin{equation*}
        \X
        = \foldOp{\foldPar}\left(\remev{f}{\widetilde{\b}}\right)
        =
        \begin{pmatrix}
        \remev{f}{\b}
        \\
        \remev{f}{\FSRSoffset\b}
        \\
        \vdots
        \\
        \remev{f}{\FSRSoffset^{\foldPar-1} \b}
       \end{pmatrix}
       \overset{(*)}{=}
       \begin{pmatrix}
        \remev{f^{(1)}}{\b}
        \\
        \remev{f^{(2)}}{\b}
        \\
        \vdots
        \\
        \remev{f^{(\foldPar)}}{\b}
       \end{pmatrix}
    \end{equation*}
    where $(*)$ follows by~\autoref{lem:alpha_power_remev}.
    By fixing a basis $\vecgamma$ of $\Fqmh$ over $\Fqm$ we can represent $\X$ over $\Fqmh$ as
    \begin{equation*}
        \extInv_{\vecgamma}(\X) = \remev{F}{\b}
    \end{equation*}
    where $F(x) = \sum_{i=0}^{n-1}F_i x^i$ with $F_i=\extInv_{\vecgamma}\left( (f_i^{(1)}, f_i^{(2)}, \dots, f_i^{(\foldPar)})^\top \right)$ for all $i = 1, \dots, n-1$.
\end{proof}

\autoref{thm:sum-rank_relation_skew_and_sum-rank_metric} shows that applying the elementwise $\Fqm$-linear transformation from~\cite[Thm.~3]{martinez2018skew} to \emph{unfolded} matrices yields an isometry between the skew metric and the sum-rank metric for matrices obtained from folded vectors.

\begin{theorem}\label{thm:sum-rank_relation_skew_and_sum-rank_metric}
 Let $\pe$ be a primitive element of $\Fqm$ and consider $\shots \in \NN^{\ast}$.
 Let $1\leq \lenShot{i} \leq m$ for all $i=1,\dots,\shots$ and let $\a=(a_1, \dots, a_\shots)\in\Fqm^\shots$ contain representatives from different conjugacy classes.
 Let the folding parameter $\foldPar$ divide $\lenShot{i}$ for all $i=1,\dots,\shots$, define the $\shots$-composition $\lenFLRSVec=(\lenFLRSshot{1}, \lenFLRSshot{2}, \dots, \lenFLRSshot{\shots})$ with $\lenFLRSshot{i}=\frac{\lenShot{i}}{\foldPar}$ for all $i=1,\dots,\shots$ and define $\h=(\foldPar, \dots, \foldPar)\in\ZZ_{\geq 0}^\shots$.
 Let $\diag(\vecbeta^{-1})$ denote the diagonal matrix of the vector
 \begin{equation}
    \vecbeta^{-1}\defeq\left(1,\pe^{-1}, \dots, (\pe^{n_1-1})^{-1} \mid \dots \mid 1, \pe^{-1}, \dots, (\pe^{n_\shots-1})^{-1}\right)
 \end{equation}
 and define the map
  \begin{align}
  \varphi_{\alpha} \, : \, \Fqm^{\foldPar \times \lenFLRS} &\to  \Fqm^{\foldPar \times \lenFLRS}, \label{eq:sum_rank_to_skew}\\
  (\shot{\X}{1} \mid \shot{\X}{2} \mid \dots \mid \shot{\X}{\ell}) &\mapsto
  \foldOp{\foldParVec}(\foldOpInv{\foldParVec}(\X) \cdot \diag(\vecbeta^{-1})). \notag
  \end{align}
  Then for any $\X \in \Fqm^{\foldPar \times \lenFLRS}$ we have that the mapping $\varphi_{\alpha}$ is an isometry between the skew metric and the sum-rank metric, i.e. we have that
  \begin{equation}
    \SkewWeight^\set{B}(\varphi_{\alpha}(\X)) = \SumRankWeight(\X)
  \end{equation}
  where $\set{B}=\{\op{a_i}{\pe^{j\foldPar}}/\pe^{j\foldPar}: j=0,\dots\lenShot{i}-1, i=1,\dots,\shots\}$.
\end{theorem}

\begin{proof}
  The vectors $\shot{\vecalpha}{i} \defeq (1, \pe^\foldPar, \dots, \pe^{(\lenFLRSshot{i}-1)\foldPar})$ contain $\Fq$-linearly independent elements from $\Fqm$ since $\pe$ is a primitive element of $\Fqm$ and $\lenShot{i}=\lenFLRSshot{i}\foldPar \leq m$ for all $i=1,\dots,\shots$.
  Thus, by~\cite[Thm.~4.5]{lam1988vandermonde} we have that the vectors
  \begin{equation*}
    \shot{\b}{i} = \left(\op{a_i}{1}, \op{a_i}{\pe^\foldPar}/\pe^\foldPar, \dots, \op{a_i}{\pe^{(\lenFLRSshot{i}-1)\foldPar}}/\pe^{(\lenFLRSshot{i}-1)\foldPar}\right)
  \end{equation*}
  contain P-independent elements for all $i=1, \dots, \shots$.
  Since $a_1, \dots, a_\shots$ are representatives of different conjugacy classes of $\Fqm$, we also have that the entries in $\b=(\shot{\b}{1} \mid \shot{\b}{2} \mid \dots \mid \shot{\b}{\shots})$ are P-independent which implies that $\set{B}$ is a P-independent set (cf.~\cite[Thm.~9]{martinez2019reliable} and~\cite{lam1985general, lam1988vandermonde}).

  By using the relation between the generalized operator evaluation and the remainder evaluation in~\eqref{eq:rel_remev_opev} and the result of~\autoref{lem:folded_mat_to_interleaved_mat_fsrs}, we can write the blocks of the transformed tuple
  \begin{equation*}
    \widetilde{\X} \defeq \varphi_{\pe}(\X)=\foldOp{\foldParVec}\left(\foldOpInv{\foldParVec}(\X)\cdot \diag(\vecbeta^{-1})\right)
  \end{equation*}
  in terms of the remainder evaluation as
  \begin{equation*}
    \shot{\widetilde{\X}}{i}
    \defeq
    \foldOp{\foldPar}\left(\foldOpInv{\foldPar}(\shot{\X}{i})\cdot \diag\left((\shot{\vecbeta}{i})^{-1}\right)\right)
    =
    \begin{pmatrix}
     \remev{f}{\shot{\b}{i}}
     \\
     \remev{f}{\FSRSoffset \shot{\b}{i}}
     \\
     \vdots
     \\
     \remev{f}{\FSRSoffset^{\foldPar - 1} \shot{\b}{i}}
    \end{pmatrix}
    =
    \begin{pmatrix}
     \remev{f^{(1)}}{\shot{\b}{i}}
     \\
     \remev{f^{(2)}}{\shot{\b}{i}}
     \\
     \vdots
     \\
     \remev{f^{(\foldPar)}}{\shot{\b}{i}}
    \end{pmatrix},
  \end{equation*}
  where $(\shot{\vecbeta}{i})^{-1} \defeq (1,\pe^{-1},\dots,(\pe^{\lenShot{i}-1})^{-1})$ for all $i=1,\dots,\shots$ such that $\vecbeta^{-1}=((\shot{\vecbeta}{1})^{-1} \mid \dots \mid (\shot{\vecbeta}{\shots})^{-1})$ and $f^{(j)}=\sum_{l=0}^{n-1} f_l \genNorm{a_i}{\FSRSoffset^j} x^l$.

  Hence we can write each transformed block $\shot{\widetilde{\X}}{i}$ over $\Fqmh$ as an evaluation of $F \in \Fqmh[x; \aut]_{<n}$ at the P-independent elements from $\Fqm$ in $\shot{\b}{i}$, i.e. we have
  \begin{equation*}
    \shot{\widetilde{\x}}{i} \defeq \extInv_{\vecgamma}(\shot{\widetilde{\X}}{i}) = \remev{F}{\shot{\b}{i}}
    \quad \forall i = 1, \dots, \shots.
  \end{equation*}

  Define the vectors $\widetilde{\x} \defeq (\shot{\widetilde{\x}}{1} \mid \shot{\widetilde{\x}}{2} \mid \dots \mid \shot{\widetilde{\x}}{\shots}) \in \Fqmh^\lenFLRS$ and $\x \defeq (\shot{\x}{1} \mid \shot{\x}{2} \mid \dots \mid \shot{\x}{\shots}) \in \Fqmh^\lenFLRS$.
  Then it follows from~\cite[Thm.~3]{martinez2018skew} that
  \begin{equation*}
   \SkewWeight^{\set{B}}(\widetilde{\x}) = \SumRankWeight(\x).
  \end{equation*}
\end{proof}

\autoref{ex:sum_rank_skew_map} illustrates the operator $\varphi_{\alpha}$.
\begin{example}\label{ex:sum_rank_skew_map}
  Consider a matrix $\X=(\shot{\X}{1} \mid\shot{\X}{2})\in\Fqm^{\foldPar \times \lenFLRS}$ where $\foldPar=3$ and $\lenFLRSVec=(2,3)$.
  Then the operator $\varphi_{\alpha}$ applied to $\X$ gives
  \begin{equation*}
    \varphi_{\alpha}(\X)=
    \left(
    \begin{array}{cc|ccc}
     \shot{x}{1}_{1,1}/1 & \shot{x}{1}_{1,2}/\pe^3 & \shot{x}{2}_{1,1}/1 & \shot{x}{2}_{1,2}/\pe^3 & \shot{x}{2}_{1,3}/\pe^6
     \\
     \shot{x}{1}_{2,1}/\pe & \shot{x}{1}_{2,2}/\pe^4 & \shot{x}{2}_{2,1}/\pe & \shot{x}{2}_{2,2}/\pe^4 & \shot{x}{2}_{2,3}/\pe^7
     \\
     \shot{x}{1}_{3,1}/\pe^2 & \shot{x}{1}_{3,2}/\pe^5 & \shot{x}{2}_{3,1}/\pe^2 & \shot{x}{2}_{3,2}/\pe^5 & \shot{x}{2}_{3,3}/\pe^8
    \end{array}
    \right).
  \end{equation*}
\end{example}

Skew Reed--Solomon codes were proposed by Boucher and Ulmer in~\cite{boucher2014linear} and further investigated in~\cite{liu2015construction, martinez2018skew}.

\begin{definition}[Skew Reed--Solomon Codes~\cite{boucher2014linear}]
    Let $\b=(b_1, b_2, \dots, b_n) \in \Fqm^n$ contain P-independent elements from $\Fqm$. Then a \emph{\acl{SRS} code} of length $n$ and dimension $k \leq n$ is defined as
    \begin{equation*}
      \skewRS{\b; n,k}=\{\remev{f}{\b} : f \in \SkewPolyringZeroDer_{<k}\}.
    \end{equation*}
\end{definition}

\begin{definition}[Folded Skew Reed--Solomon Codes]\label{def:folded_srs}
    Let $\pe$ be a primitive element of $\Fqm$ and define $\FSRSoffset \defeq \aut(\pe)/\pe$.
    Let $a_1,\dots,a_\shots$ be representatives of pairwise distinct nontrivial conjugacy classes of $\Fqm$.
    Define $\b=(\shot{\b}{1} \mid \shot{\b}{2} \mid \dots \mid \shot{\b}{\shots})$ with
    \begin{equation*}
        \shot{\b}{i} \defeq
        a_i \left(1, \FSRSoffset, \FSRSoffset^2, \dots, \FSRSoffset^{\lenShot{i}-1}\right)
        \in \Fqm^{\lenShot{i}}
        \quad \forall i = 1, \dots, \shots.
    \end{equation*}
    Choose a folding parameter $\foldPar \in \NN$ satisfying $\foldPar \mid \lenShot{i}$ for all $1 \leq i \leq \shots$ and $\lenFLRSshot{i} \defeq \frac{\lenShot{i}}{\foldPar} \leq \foldPar$ for all $1 \leq i \leq \shots$ and write
    $\N \defeq (\lenFLRSshot{1}, \dots, \lenFLRSshot{\shots})$.
    Then an \emph{$\foldPar$-folded \acl{SRS} code} of length $\lenFLRS \defeq \sum_{i=1}^{\shots} \lenFLRSshot{i}$ and dimension $k$ is defined as
    \begin{equation}
        \foldedSkewRS{\b, \foldPar; \lenFLRSVec, k}
        \defeq
        \left\{
        \foldOp{\foldPar}(\remev{f}{\b}) : f \in \SkewPolyringZeroDer_{<k} \right\}.
    \end{equation}
\end{definition}

\begin{remark}
  Equivalently, we can define \ac{FSRS} codes as
  \begin{equation}
        \foldedSkewRS{\b, \foldPar; \lenFLRSVec, k}
        \defeq
        \left\{
        \foldOp{\foldPar}(\c) : \c \in \skewRS{\b; n,k} \right\}
  \end{equation}
  where $\b$ is defined as in~\autoref{def:folded_srs}.
\end{remark}

Note that any codeword $\C \in \foldedSkewRS{\b, \foldPar; \lenFLRSVec, k} \subseteq
\Fqm^{\foldPar \times \lenFLRS}$ corresponding to a message polynomial $f \in \SkewPolyringZeroDer_{<k}$ has the form
\begin{equation*}
    \C = \left( \C^{(1)} \mid \dots \mid \C^{(\shots)} \right)
\end{equation*}
where
\begin{equation*}
    \C^{(i)} =
    \begin{pmatrix}
        \remev{f}{a_i} & \remev{f}{\FSRSoffset^\foldPar a_i} & \cdots
        & \remev{f}{\FSRSoffset^{\lenShot{i}-\foldPar} a_i}
        \\
        \remev{f}{\FSRSoffset a_i} & \remev{f}{\FSRSoffset^{\foldPar+1} a_i} & \cdots
        & \remev{f}{\FSRSoffset^{\lenShot{i}-\foldPar+1} a_i}
        \\
        \vdots & \vdots & \ddots & \vdots
        \\
        \remev{f}{\FSRSoffset^{\foldPar - 1} a_i} & \remev{f}{\FSRSoffset^{2\foldPar-1} a_i} & \cdots
        & \remev{f}{\FSRSoffset^{\lenShot{i}-1} a_i}
    \end{pmatrix}
    \in \Fqm^{\foldPar \times \lenFLRSshot{i}}
\end{equation*}
for all $i \in \{1, \ldots, \shots\}$.

\begin{proposition}[Relation between \ac{FLRS} and \ac{FSRS} Codes]\label{prop:rel_flrs_fsrs}
  Let $\foldedSkewRS{\b, \foldPar; \allowbreak \lenFLRSVec, k}$ be an \ac{FSRS} code whose parameters comply with~\autoref{def:folded_srs}.
  Then,
  \begin{equation*}
    \foldedSkewRS{\b, \foldPar; \lenFLRSVec, k}
    =
    \{\varphi_{\pe}(\C) : \C \in \foldedLinRS{\pe, \a, \foldParVec; \lenFLRSVec, k}\}
  \end{equation*}
  where the code $\foldedLinRS{\pe, \a, \foldParVec; \lenFLRSVec, k}$ with $\foldParVec = (\foldPar, \dots, \foldPar)$ is considered as subset of $\Fqm^{\foldPar \times \lenFLRS}$.
\end{proposition}

\begin{proof}
  The result follows directly by using the relation between the generalized operator evaluation and the remainder evaluation in~\eqref{eq:rel_remev_opev}.
\end{proof}

By using the isometry between the sum-rank metric and the skew metric from \autoref{thm:sum-rank_relation_skew_and_sum-rank_metric}, we obtain the following corollary from~\autoref{thm:minimum_distance_flrs}.

\begin{corollary}[Minimum Skew Distance]
  The minimum skew distance of an \ac{FSRS} code $\mycode{C} \defeq \foldedSkewRS{\b, \foldPar; \lenFLRSVec, k}$ of length $\lenFLRS=\sum_{i=1}^{\shots}\lenFLRSshot{i}$ as defined in~\autoref{def:folded_srs} is
  \begin{equation*}
    \SkewDist(\mycode{C}) = \lenFLRS - \left\lceil\frac{k}{\foldPar}\right\rceil + 1.
  \end{equation*}
\end{corollary}

\subsection{Interpolation-Based Decoding of Folded Skew Reed--Solomon Codes} \label{subsec:decoding-of-fsrs-codes}

We now consider interpolation-based decoding of \ac{FSRS} codes with respect to the skew metric.
As a channel model we consider the skew error channel with input and output alphabet $\Fqm^{\foldPar \times \lenFLRS}$, where the input $\C$ is related to the output $\R$ by
\begin{equation}\label{eq:skew_error_channel}
  \R = \C + \E
\end{equation}
and $\E$ with $\SkewWeight(\E)=t$ is chosen uniformly at random from all matrices from $\Fqm^{\foldPar \times \lenFLRS}$ with skew weight $t$.

Suppose we transmit a codeword $\C \in \foldedSkewRS{\b, \foldPar; \lenFLRSVec, k}$ over a skew error channel~\eqref{eq:skew_error_channel} and receive a matrix $\R=(\shot{\R}{1} \mid \shot{\R}{2} \mid \dots \mid \shot{\R}{\shots})$.
Let $\varphi^{-1}_{\alpha}$ denote the inverse map of $\varphi_{\alpha}$.
By using the isometry between the sum-rank metric and the relation between \ac{FLRS} codes and \ac{FSRS} codes, we can transform the received matrix $\R$ to
\begin{equation*}
  \widetilde{\R} \defeq \varphi^{-1}_{\alpha}(\R) = \varphi^{-1}_{\alpha}(\C) + \varphi^{-1}_{\alpha}(\E)
\end{equation*}
where $\varphi^{-1}_{\alpha}(\C)$ is in the corresponding \ac{FLRS} code $\foldedLinRS{\pe, \a, \foldParVec; \lenFLRSVec, k}$ (see \autoref{prop:rel_flrs_fsrs}) and $\varphi^{-1}_{\alpha}(\E)$ has sum-rank weight $t$ (see~\autoref{thm:sum-rank_relation_skew_and_sum-rank_metric}).
Hence, the decoding problem for \ac{FSRS} codes with respect to the skew metric is mapped to an equivalent decoding problem for \ac{FLRS} codes in the sum-rank metric.

Therefore, we can use the interpolation-based decoding schemes for \ac{FLRS} codes from~\autoref{sec:decoding} to decode \ac{FSRS} codes in the skew metric as follows:

\begin{enumerate}
  \item Compute $\widetilde{\R} \defeq \varphi^{-1}_{\alpha}(\R)$, which requires $\OCompl{n}$ operations in $\Fqm$.

  \item Apply a decoder for \ac{FLRS} codes in the sum-rank metric (e.g.~\autoref{alg:decoder}) to $\widetilde{\R}$.
\end{enumerate}

\section{Conclusion} \label{sec:conclusion}

\acresetall

We introduced the family of \ac{FLRS} codes whose members are \ac{MSRD} codes for appropriate parameter choices.
We further described an efficient decoding scheme to correct sum-rank errors in the context of list and probabilistic
unique decoding with quadratic complexity in the length of the unfolded code.
Up to our knowledge, this is the first explicit \ac{MSRD} code construction that allows different block sizes and has
an explicit efficient decoding algorithm.
We analyzed the decoder and gave upper bounds on both list size and failure probability.
Monte Carlo simulations verified that the observed failure probability is indeed below the derived bound and further
experiments show that the assumption under which the upper bound was derived is reasonable.
Since the proposed decoding scheme has a rate restriction, we investigated a Justesen-like improvement tailored to
high-rate \ac{FLRS} codes.

The focus of the second part of the paper was the skew metric for which we introduced \ac{FSRS} codes in the
zero-derivation setting.
Moreover, we explained how the decoding scheme for \ac{FLRS} codes in the sum-rank metric can be applied to the
presented skew-metric codes.

Goals for further research could be the extension of \ac{FSRS} codes to the nonzero-derivation case or to more general
parameters as e.g.\ code locators.
Moreover, it is tempting to study if there are other useful ways of folding codes in different metrics.

\vspace*{.5cm}
\noindent
\textbf{Acknowledgments.}
F. Hörmann and H. Bartz acknowledge the financial support by the Federal Ministry of Education and Research of
Germany in the programme of ``Souverän. Digital. Vernetzt.'' Joint project 6G-RIC, project identification
number: 16KISK022.

\bibliographystyle{splncs04}
\bibliography{references}

\end{document}